\documentclass[11pt,a4paper]{article}
\usepackage{cite}
\usepackage[utf8]{inputenc}
\usepackage{amsmath,amssymb,amsthm}
\usepackage{tikz}
\usetikzlibrary{decorations.pathreplacing}
\usepackage{enumitem}
\usepackage[margin=2.6cm]{geometry}
\usepackage[colorlinks=true,linkcolor=blue!60!black,citecolor=blue!60!black,urlcolor=blue!60!black]{hyperref}
\usepackage{xcolor}

\theoremstyle{plain}
\newtheorem{theorem}{Theorem}[section]
\newtheorem{proposition}[theorem]{Proposition}
\newtheorem{lemma}[theorem]{Lemma}
\newtheorem{corollary}[theorem]{Corollary}
\theoremstyle{definition}
\newtheorem{example}[theorem]{Example}
\theoremstyle{remark}
\newtheorem{remark}[theorem]{Remark}

\newcommand{\C}{\mathbb{C}}
\newcommand{\gl}{\mathfrak{gl}}
\newcommand{\B}{\mathcal{B}}
\newcommand{\Tr}{\operatorname{Tr}}
\newcommand{\spec}{\operatorname{spec}}
\newcommand{\ip}[2]{\langle #1,#2\rangle}
\newcommand{\lp}{\lambda^{+}}
\newcommand{\Cd}{\C^d}
\newcommand{\Cdbar}{\overline{\C^d}}
\newcommand{\one}{\mathbf{1}}
\newcommand{\mj}{\prec}

\title{Minimal output entropy for the channels\\
that add or remove a box of a Young diagram,\\
and a Pauli principle for every permutation symmetry}
\author{Robin Reuvers\\[4pt]
\small Dipartimento di Matematica e Fisica, Universit\`a degli Studi Roma Tre,\\[-2pt]
\small Largo San Leonardo Murialdo 1, 00146 Roma, Italy\\[-2pt]
\small \texttt{robin.reuvers@uniroma3.it}}
\date{}

\begin{document}
\maketitle

\begin{abstract}
Two irreducible representations of $U(d)$ whose Young diagrams differ by a single box are connected by four covariant quantum channels: the box can be removed or added, and one keeps either the new diagram or the box ($\Cd$ or its dual). We prove that for each of these channels the output of a coherent state majorizes every other output, so that coherent states minimize the output entropy, and we compute the optimal output explicitly in terms of hook lengths. For the two channels that keep the box we also determine all minimizers; these need not be coherent, even when they are pure. As an application we consider $N$ particles with a given permutation symmetry and find sharp bounds on the spectrum of the reduced density matrix of a single particle, as well as the exact set of spectra that mixed states can reach (for fermions these are the Pauli principle and Coleman's theorem).
\end{abstract}



\section{Introduction}\label{sec:intro}

Wehrl~\cite{Wehrl} proposed to measure how classical a quantum state is by the entropy of its Husimi function, and conjectured that, among all states of a particle on the line, the Glauber coherent states minimize it. Lieb~\cite{Lieb78} proved this, and conjectured the same for the Bloch coherent states of a spin $J$. That conjecture stood for more than thirty years. The proof by Lieb and Solovej~\cite{LS14} runs through a quantum channel,
\begin{equation}\label{eq:LS}
\Lambda(\rho)=\frac{2J+1}{2K+1}\,P_K\,(\rho\otimes\one_{K-J})\,P_K,
\end{equation}
where $P_K$ projects $\mathcal H_J\otimes\mathcal H_{K-J}$ onto the subspace of largest total spin. (The map had appeared before, in~\cite{LS91} and as the optimal cloning channel for qubits~\cite{Werner98,KeylWerner99}.) Lieb and Solovej showed that
\begin{equation}\label{eq:LSineq}
\Tr f\big(\Lambda(\rho)\big)\ \geq\ \Tr f\big(\Lambda(|\Omega\rangle\langle\Omega|)\big)
\end{equation}
for every state $\rho$, every coherent state $|\Omega\rangle$ and every concave function $f$. In other words, the output of a coherent state \emph{majorizes} the output of every other state. As $K\to\infty$, the output of $\Lambda$ turns into the Husimi function, and Lieb's conjecture follows.

The channel \eqref{eq:LS} has natural analogues for other groups and representations, and one may ask whether \eqref{eq:LSineq} extends to them. In the semiclassical limit, where the output becomes a Husimi function, this question is by now well understood. After~\cite{Schupp99,GZ01,Sugita02,Bodmann04}, the $SU(1,1)$ case was raised in~\cite{LS21} and settled in~\cite{Kulikov22}, and optimizers and stability were studied in~\cite{Carlen91,Frank23,KNOT25,FNT25,NRT26}. Recently, Zhang~\cite{Zhang26} proved the entropy bound, and the bounds on all moments of order $p>1$, for every compact connected semisimple Lie group, with coherent states as the only optimizers. The channels themselves are harder. When \eqref{eq:LSineq} holds along a family of channels whose outputs tend to a Husimi function, the semiclassical inequality follows in the limit, as Lieb's conjecture did above; no converse is known. For covariant channels between irreducible representations of $U(d)$, \eqref{eq:LSineq} had been proved, to our knowledge, only in trivial cases and for two families. The trivial cases include channels out of $\Cd$ and channels whose coherent-state output is pure. The two families are the bosonic channels of Lieb and Solovej~\cite{LS16}, which enlarge a symmetric representation, and the fermionic channels that add or remove one particle, where \eqref{eq:LSineq} is the Pauli principle. In both families the representations are symmetric or antisymmetric powers of $\Cd$, up to duals and powers of the determinant. For all other Young diagrams, no nontrivial majorization result was known, even though covariant channels have been studied closely: Al Nuwairan~\cite{AlNuwairan14} determined the extreme points for $SU(2)$, Aschieri, Ruba and Solovej~\cite{ARS25} analysed them semiclassically, and Man\v{c}inska and Theil~\cite{MT25} classified them for $SU(d)$.

A natural family to test is that of the \emph{Cartan channels}. If $V_\lambda$ and $V_\alpha$ are irreducible representations of $U(d)$, with dominant weights $\lambda$ and $\alpha$, then $V_{\lambda+\alpha}$ occurs exactly once in $V_\lambda\otimes V_\alpha$, and
\begin{equation}\label{eq:cartan}
\Lambda_{\lambda,\alpha}(\rho):=\frac{\dim V_\lambda}{\dim V_{\lambda+\alpha}}\,P_{\lambda+\alpha}(\rho\otimes\one)P_{\lambda+\alpha},
\end{equation}
with $P_{\lambda+\alpha}$ the projection onto $V_{\lambda+\alpha}$, is a covariant channel from $\B(V_\lambda)$ to $\B(V_{\lambda+\alpha})$; \eqref{eq:LS} is the case $d=2$. In his thesis, Aschieri calls them minimal channels and asks whether they all satisfy \eqref{eq:LSineq}~\cite[\S2.4, eq.~(2.130)]{Aschieri25}. The results of Lieb and Solovej answer this positively for $d=2$~\cite{LS14} and when $\lambda$ and $\alpha$ are single rows~\cite{LS16}.

\medskip
In this paper we take the smallest possible step between two representations of $U(d)$: their Young diagrams differ by a single box. Let $\lp$ be a Young diagram with at most $d$ rows, let $(j,m)$ be a box that can be removed from it, and let $\lambda=\lp-e_j$ be the diagram without that box. By Pieri's rule, $V_{\lp}$ occurs exactly once in $V_\lambda\otimes\Cd$, so there is an isometric intertwiner
\begin{equation}\label{eq:T}
T:V_{\lp}\to V_\lambda\otimes\Cd,
\end{equation}
unique up to a phase. Applying $T$ to a state and tracing out one of the two factors gives two channels that \emph{remove} the box: one keeps the diagram, so that the output lives on $V_\lambda$, and the other keeps the removed particle, so that the output lives on $\Cd$. Similarly, since $V_\lambda$ occurs exactly once in $V_{\lp}\otimes\overline{\Cd}$, there is an isometric intertwiner $V_\lambda\to V_{\lp}\otimes\overline{\Cd}$, and it gives two channels that \emph{add} the box. All four are covariant. 

Our main results, Theorems~\ref{thm:main} and~\ref{cor:adding}, say that \eqref{eq:LSineq} holds for all four channels, and give the coherent-state output explicitly. For the channel that removes the box and keeps the particle, the sum of the $k$ largest eigenvalues of the coherent-state output is
\begin{equation}\label{eq:hookintro0}
s_k=\prod_{i=1}^{j-k}\Big(1-\frac1{h_{(i,m)}}\Big),
\end{equation}
where $h_{(i,m)}$ is the hook length of the box $(i,m)$ of $\lp$. That is, the answer is read off the column of the removed box: for $k=1$ one multiplies $1-1/h$ over the boxes above it, and each further unit of $k$ drops the lowest factor that is left (Figure~\ref{fig:hooks}). Previously, this was known for a single row, where the channels are the easiest case of~\cite{LS16}; for a single column, where \eqref{eq:LSineq} is the Pauli principle; for the largest output eigenvalue, $k=1$~\cite{Reuvers19}; and, through the spectral computation of~\cite{RGKZ26}, for the smallest one, which carries information only when the box is removed from row $d$. When the box goes into the first row, or comes out of row $d$, the channel that keeps the diagram is the Cartan channel \eqref{eq:cartan} with $V_\alpha=\Cd$ or $\Cdbar$, so our results answer Aschieri's question positively for these two $\alpha$. For general $\alpha$ the answer is negative (Remark~\ref{rem:cartan}).

Two consequences follow immediately. Coherent states minimize the von Neumann and all R\'enyi output entropies of the four channels, and this determines their Holevo capacities (Corollary~\ref{cor:entropy}). Some of these channels are familiar: optimal cloning of one additional copy~\cite{Werner98} is the channel that adds a box to a single row and keeps the diagram, and its complementary channel, which keeps the particle, is optimal transposition into a single copy~\cite{BGSQ26}; optimal purity amplification~\cite{LFIC25} acts, within each Schur--Weyl block, as a channel that removes a box and keeps the particle (Remark~\ref{rem:occur}). For the two channels that keep the particle we moreover find \emph{all} minimizers of the output entropy (Theorem~\ref{thm:minimizers}). They need not be coherent, even when they are pure.

\medskip
The second theme of this paper is the Pauli principle, and what replaces it for particles that are neither bosons nor fermions. By Schur--Weyl duality, a Young diagram $\nu$ with $N$ boxes labels an irreducible representation $V_\nu$ of $U(d)$ and an irreducible representation $S_\nu$ of the symmetric group $S_N$, and
\[
V_\nu\otimes S_\nu\subset\otimes^N\Cd
\]
is the space of states of $N$ particles with permutation symmetry $\nu$. For a single column it is the space of $N$ fermions, and the Pauli principle says that the reduced density matrix $\gamma$ of one particle satisfies $\gamma\leq\one/N$. Coleman~\cite{Coleman63} made this the starting point of the $N$-representability problem, which asks which reduced density matrices come from $N$-particle states. For pure states, the spectrum of $\gamma$ satisfies further linear constraints, which Altunbulak and Klyachko~\cite{AK08} listed explicitly in low rank. In~\cite{Reuvers19} we asked what replaces $1/N$ for a general diagram $\nu$, and found that the largest eigenvalue of $\gamma$ is at most
\begin{equation}\label{eq:hook1}
\max_{(j,m)}\ \prod_{i<j}\Big(1-\frac{1}{h_{(i,m)}}\Big),
\end{equation}
where the maximum runs over the removable boxes of $\nu$, with equality for a coherent state of $V_\nu$ tensored with a suitable basis vector of $S_\nu$. The bound \eqref{eq:hook1} is the largest eigenvalue of the operator $\one^{\otimes(N-1)}\otimes|e_1\rangle\langle e_1|$ compressed to $V_\nu\otimes S_\nu$. Rico, Grinko, Krebs and Zaw~\cite{RGKZ26} recently met this operator when designing entanglement witnesses. They computed its full spectrum in a Gelfand--Tsetlin basis (for such bases see also~\cite{GBO23}) and used it in their Theorem~1; its largest eigenvalue is \eqref{eq:hook1}~\cite{Reuvers19}.

The link with the channels is the following. Restricted to the permutations of the first $N-1$ particles, $S_\nu$ splits into one summand for each removable box $(j,m)$ of $\nu$, and the states whose $S_\nu$-part lies in that summand are those whose first $N-1$ particles have symmetry $\nu-e_j$. For these states, tracing out the first $N-1$ particles is exactly the channel that removes the box $(j,m)$ and keeps the particle~\cite{Reuvers19}. Our majorization theorem therefore gives them a complete Pauli principle: $s_k(\gamma)$ is at most the product \eqref{eq:hookintro0}, with $\lp=\nu$, for every $k$, and, for mixed states, these are the only constraints on $\gamma$ (Theorem~\ref{thm:nested}). For fermions there is only one box, the product telescopes to $k/N$, and this is the Pauli principle together with Coleman's theorem~\cite{Coleman63}.

A general state of symmetry $\nu$ mixes the boxes, and its $\gamma$ is a convex combination of reduced density matrices of the previous kind. So \eqref{eq:hook1} extends from the largest eigenvalue to the sum $s_k(\gamma)$ of the $k$ largest eigenvalues, simply by dropping the $k-1$ hook factors nearest the removed box: for every state,
\begin{equation}\label{eq:hookintro}
s_k(\gamma)\ \leq\ \max_{(j,m)}\ \prod_{i=1}^{j-k}\Big(1-\frac{1}{h_{(i,m)}}\Big),
\end{equation}
and each of these bounds is attained (Corollary~\ref{thm:rdm}). For $k\geq2$ the bounds are new. They are again the largest eigenvalues of explicit operators, but unlike for $k=1$, symmetry no longer diagonalizes these operators, because the relevant branching has multiplicities (Remark~\ref{rem:grinko}). The maximizing box may depend on $k$. For $\nu=(5,5,5,4)$, for example, the bound on $s_1(\gamma)$ comes from the box in the fourth row and the bound on $s_2(\gamma)$ from the box in the third row, and no single $\gamma$ attains both. Corollary~\ref{thm:polytope} then determines exactly which $\gamma$ mixed states can reach: those whose spectra lie in the polytope spanned by the permutations of one vector for each removable box, whose partial sums are the products in \eqref{eq:hookintro}.

\medskip
The proof stays within the representation theory of $U(d)$. Adding a box is easier to handle than removing one, and passing to dual representations exchanges the two, so we prove the adding case. There, the eigenvalue sums we need to bound are norms of compressions of a projection onto an isotypic component. Such a projection is dominated by a Casimir-type operator, and the compressions of that operator are controlled by branching to $U(k)$. This gives a sharp bound for diagrams of a special shape, and two reductions bring every diagram into that shape: the columns to the right of the box can be factored out exactly, and the tallest columns, which reach far below the box, can be deleted at the price of an inequality in the right direction. The strategy is described in more detail at the start of Section~\ref{sec:proof}.

\section{Main results}\label{sec:main}

\subsection{Notation}\label{sec:notation}

Let $e_1,\dots,e_d$ be the standard basis of $\Cd$, and let $E_{ab}:=|e_a\rangle\langle e_b|$ be the matrix units, which form a basis of $\gl_d$. For $n\leq d$ we regard $U(n)\subset U(d)$ and $\gl_n\subset\gl_d$ as acting on $\C^n=\operatorname{span}(e_1,\dots,e_n)$ and trivially on $e_{n+1},\dots,e_d$.

A Young diagram is a sequence of integers $\mu_1\geq\dots\geq\mu_d\geq0$, drawn with $\mu_a$ boxes in row $a$. The box in row $a$ and column $c$ is the box $(a,c)$. Its hook length $h_{(a,c)}$ counts the boxes to its right, the boxes below it, and the box itself. A box is \emph{removable} if deleting it leaves a Young diagram; such a box is the last one in its row and in its column, and its hook length is $1$.

We write $V_\mu$ for the irreducible representation of $U(d)$ with highest weight $\mu$, $\pi_\mu$ for the representation map, of $U(d)$ and of $\gl_d$, and $v_\mu$ for a unit highest weight vector; we often write $Xv$ for $\pi_\mu(X)v$. Two standard facts will be used throughout: $V_\mu$ is spanned by the vectors $E_{a_1b_1}\cdots E_{a_rb_r}v_\mu$ with $a_s>b_s$, and the weights of $V_\mu$ lie in the convex hull of the $S_d$-orbit of $\mu$, so that every weight $\theta$ satisfies $0\leq\theta_a\leq\mu_1$. The \emph{coherent states} of $V_\mu$ are the vectors in the $U(d)$-orbit of $v_\mu$, and the corresponding projections~\cite{Perelomov}. We write $D_n(\mu)$ for the dimension of the irreducible representation of $U(n)$ with highest weight $(\mu_1,\dots,\mu_n)$, so that $D_d(\mu)=\dim V_\mu$.

For a vector $x$, $s_k(x)$ is the sum of its $k$ largest entries, or of all entries if there are fewer than $k$. For a Hermitian matrix $A$, $\spec A$ is the vector of its eigenvalues in decreasing order, and $s_k(A):=s_k(\spec A)$. We write $x\mj y$, and say that $x$ is majorized by $y$, if the entries of $x$ and $y$ have the same sum and $s_k(x)\leq s_k(y)$ for all $k$, and $A\mj B$ if $\spec A\mj\spec B$~\cite{MOA}. Empty products are $1$.

Throughout the paper, $\lp$ is a Young diagram with at most $d$ rows, $(j,m)$ is a removable box of $\lp$, so that $m=\lp_j$, and
\[
\lambda:=\lp-e_j
\]
is the diagram without that box. The coherent states that everything will be compared against are
\[
\rho_0:=|v_{\lp}\rangle\langle v_{\lp}|\ \text{ on }V_{\lp},\qquad \sigma_0:=|v_\lambda\rangle\langle v_\lambda|\ \text{ on }V_\lambda.
\] 

\subsection{The four channels}\label{sec:channels}
By Pieri's rule, $V_{\lp}$ occurs exactly once in $V_\lambda\otimes\Cd$, so there is an isometric intertwiner $T:V_{\lp}\to V_\lambda\otimes\Cd$, unique up to a phase. Let $T_a:=(\one\otimes\langle e_a|)T:V_{\lp}\to V_\lambda$ be its components, so that
\begin{equation}\label{eq:Tsum}
T\phi=\sum_aT_a\phi\otimes e_a,\qquad\sum_aT_a^\dagger T_a=\one_{V_{\lp}},\qquad TT^\dagger=P,
\end{equation}
where $P$ is the projection of $V_\lambda\otimes\Cd$ onto its unique summand of type $\lp$. Applying $T$ to a state and tracing out one of the two factors gives the two channels that remove the box,
\begin{equation}\label{eq:removing}
\Phi_-(\rho):=\Tr_{\Cd}\big(T\rho T^\dagger\big)=\sum_aT_a\rho T_a^\dagger,\qquad
\Psi_-(\rho):=\Tr_{V_\lambda}\big(T\rho T^\dagger\big),
\end{equation}
from $\B(V_{\lp})$ to $\B(V_\lambda)$ and to $\B(\Cd)$, respectively. The first keeps the diagram, the second keeps the removed particle. The matrix elements of the second are $\langle e_a|\Psi_-(\rho)|e_b\rangle=\Tr(\rho\,T_b^\dagger T_a)$.

To add the box, we use the adjoint of $T$. The operator $\Tr_{\Cd}P=\sum_aT_aT_a^\dagger$ on $V_\lambda$ commutes with $U(d)$ and has trace $\dim V_{\lp}$, so Schur's lemma gives
\begin{equation}\label{eq:schur}
\sum_{a=1}^dT_aT_a^\dagger=\kappa_d\,\one_{V_\lambda},\qquad \kappa_d:=\frac{\dim V_{\lp}}{\dim V_\lambda}.
\end{equation}
Consequently
\begin{equation}\label{eq:S}
S:V_\lambda\to V_{\lp}\otimes\Cdbar,\qquad S\phi:=\kappa_d^{-1/2}\sum_aT_a^\dagger\phi\otimes e_a
\end{equation}
is an isometry. It is an intertwiner if we let $U(d)$ act on $\Cdbar$ by $U\mapsto\overline U$, as one sees by taking adjoints in the intertwining relation for $T$. The two channels that add the box are
\begin{equation}\label{eq:adding}
\Phi_+(\sigma):=\Tr_{\Cdbar}\big(S\sigma S^\dagger\big)=\kappa_d^{-1}\sum_aT_a^\dagger\sigma T_a,\qquad
\Psi_+(\sigma):=\Tr_{V_{\lp}}\big(S\sigma S^\dagger\big),
\end{equation}
from $\B(V_\lambda)$ to $\B(V_{\lp})$ and to $\B(\Cdbar)$. Table~\ref{tab:channels} summarizes the notation: $\Phi$ keeps the diagram, $\Psi$ keeps the particle, and the sign says whether the box is removed or added.

\begin{table}[t]
\centering
\renewcommand{\arraystretch}{1.3}
\begin{tabular}{l|cc}
 & removes the box & adds the box\\ \hline
keeps the diagram & $\Phi_-:\B(V_{\lp})\to\B(V_\lambda)$ & $\Phi_+:\B(V_\lambda)\to\B(V_{\lp})$\\
keeps the particle & $\Psi_-:\B(V_{\lp})\to\B(\Cd)$ & $\Psi_+:\B(V_\lambda)\to\B(\Cdbar)$
\end{tabular}
\caption{The four channels.}
\label{tab:channels}
\end{table}

Three comments on these definitions. First, if we identify $V_{\lp}$ with the range of $P$, then $\Phi_+(\sigma)=\kappa_d^{-1}P(\sigma\otimes\one)P$. For $j=1$ this is the Cartan channel $\Lambda_{\lambda,\alpha}$ of \eqref{eq:cartan} with $V_\alpha=\Cd$. If moreover $\lambda$ is a single row, this is the channel of Lieb and Solovej~\cite{LS16} that adds one box to a symmetric representation, and for $d=2$ this is \eqref{eq:LS} with $K=J+\tfrac12$. Second, the two channels in each column of Table~\ref{tab:channels} are complementary: they come from the same isometry, $T$ or $S$, by tracing out one factor or the other. On pure inputs, complementary channels have outputs with the same nonzero eigenvalues~\cite{DevetakShor05,KMNR07}. Third, all four channels are covariant; for instance, $\Psi_-\big(\pi_{\lp}(U)\rho\,\pi_{\lp}(U)^\dagger\big)=U\,\Psi_-(\rho)\,U^\dagger$ for $U\in U(d)$, and similarly for the others, with $\overline U$ acting on $\Cdbar$. So every coherent state gives the same output spectrum, and it suffices to state the results for $\rho_0$ and $\sigma_0$.

\subsection{Removing a box}\label{sec:removing}

\begin{theorem}[Removing a box]\label{thm:main}
For every density matrix $\rho$ on $V_{\lp}$,
\[
\Phi_-(\rho)\mj\Phi_-(\rho_0),\qquad \Psi_-(\rho)\mj\Psi_-(\rho_0).
\]
The coherent-state output $\Psi_-(\rho_0)$ is diagonal in the basis $(e_a)$, with $j$ nonzero eigenvalues, and
\begin{equation}\label{eq:hooks}
s_k\big(\Psi_-(\rho_0)\big)=\prod_{i=1}^{j-k}\Big(1-\frac{1}{h_{(i,m)}}\Big),\qquad 1\leq k\leq j,
\end{equation}
where the hook lengths are those of $\lp$. The output $\Phi_-(\rho_0)$ has the same nonzero eigenvalues.
\end{theorem}

The theorem is proved in Section~\ref{sec:proof}, where we also show that the diagonal entries of the coherent-state output are
\begin{equation}\label{eq:entries}
\|T_av_{\lp}\|^2=\begin{cases}\dfrac{1}{h_{(a,m)}}\displaystyle\prod_{i<a}\Big(1-\frac{1}{h_{(i,m)}}\Big), & a\leq j,\\[10pt] 0, & a>j,\end{cases}
\end{equation}
where $h_{(j,m)}=1$ because the box is removable. These entries do not decrease with $a$: the ratio of two consecutive entries is $(h_{(a-1,m)}-1)/h_{(a,m)}\geq1$, since hook lengths drop by at least one per step down a column. So the largest eigenvalue belongs to $e_j$, the second largest to $e_{j-1}$, and so on, and the sum of the $k$ largest entries telescopes to \eqref{eq:hooks}. Figure~\ref{fig:hooks} shows an example. We write $p^{(j,m)}$ for the entries \eqref{eq:entries} in decreasing order, padded with zeros to a vector of length~$d$; it is the spectrum of $\Psi_-(\rho_0)$.

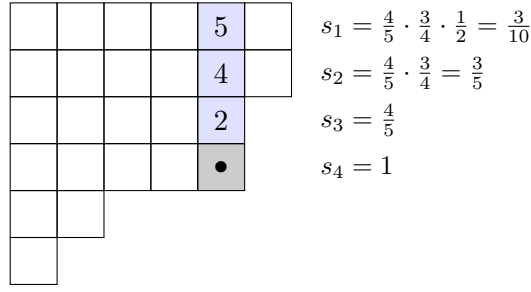
\begin{figure}[t]
\centering
\begin{tikzpicture}[scale=0.62]
\fill[blue!12] (4,0) rectangle (5,-1);
\fill[blue!12] (4,-1) rectangle (5,-2);
\fill[blue!12] (4,-2) rectangle (5,-3);
\fill[black!20] (4,-3) rectangle (5,-4);
\foreach \r/\len in {1/6,2/6,3/5,4/5,5/2,6/1}{
  \foreach \c in {1,...,\len}{ \draw (\c-1,-\r+1) rectangle (\c,-\r); }
}
\node at (4.5,-0.5) {$5$};
\node at (4.5,-1.5) {$4$};
\node at (4.5,-2.5) {$2$};
\node at (4.5,-3.5) {$\bullet$};
\node[anchor=west] at (6.4,-0.5) {\small $s_1=\tfrac45\cdot\tfrac34\cdot\tfrac12=\tfrac{3}{10}$};
\node[anchor=west] at (6.4,-1.5) {\small $s_2=\tfrac45\cdot\tfrac34=\tfrac35$};
\node[anchor=west] at (6.4,-2.5) {\small $s_3=\tfrac45$};
\node[anchor=west] at (6.4,-3.5) {\small $s_4=1$};
\end{tikzpicture}
\caption{Theorem~\ref{thm:main} for $\lp=(6,6,5,5,2,1)$ and the removable box $(4,5)$, marked $\bullet$. The hook lengths of the boxes above it are $5,4,2$. The coherent-state output of $\Psi_-$ has spectrum $(\tfrac3{10},\tfrac3{10},\tfrac15,\tfrac15)$.}
\label{fig:hooks}
\end{figure}

\subsection{Adding a box}\label{sec:adding}

Theorem~\ref{thm:main} is equivalent to the following statement about adding a box, and it is that statement we shall actually prove. The two are linked as follows. Fix $R\geq\lp_1$ and complement every diagram in a $d\times R$ rectangle,
\[
\mu\mapsto\mu^c:=(R-\mu_d,\dots,R-\mu_1).
\]
Removing the box $(j,m)$ from $\lp$ then becomes adding a box in row $d+1-j$ to $(\lp)^c$ (Figure~\ref{fig:duality}). Since the conjugate representation $\overline{V_\mu}$, which is the dual of $V_\mu$, has highest weight $(-\mu_d,\dots,-\mu_1)$ and is therefore isomorphic to $V_{\mu^c}$ up to a determinant twist~\cite[Lecture~15]{FultonHarris}, conjugating the isometry \eqref{eq:T} turns $\Phi_-$ and $\Psi_-$ into the channels $\Phi_+$ and $\Psi_+$ of the dual pair $\big((\lp)^c,\lambda^c\big)$, and output spectra and coherent states are preserved (Lemma~\ref{lem:complement}).

\begin{figure}[t]
\centering
\begin{tikzpicture}[scale=0.55]
\foreach \r/\len in {1/6,2/6,3/5,4/5,5/2,6/1}{
  \foreach \c in {1,...,6}{
    \ifnum\c>\len \fill[black!15] (\c-1,-\r+1) rectangle (\c,-\r); \fi
    \draw (\c-1,-\r+1) rectangle (\c,-\r);
  }
}
\draw[very thick] (0,0) rectangle (6,-6);
\node at (4.5,-3.5) {$\bullet$};
\end{tikzpicture}
\caption{Duality for $\lp=(6,6,5,5,2,1)$ and $d=R=6$. The complement of $\lp$ in the rectangle (shaded), rotated by $180^\circ$, is $(\lp)^c=(5,4,1,1,0,0)$. Removing the box $(4,5)$ (marked $\bullet$) from $\lp$ adds it to the complement, as the box $(3,2)$ of $\lambda^c=(5,4,2,1,0,0)$.}
\label{fig:duality}
\end{figure}
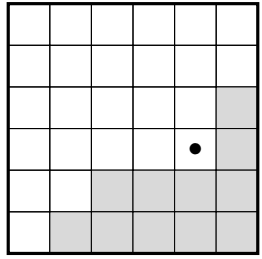

For the channels that add a box, the natural form of the answer is a ratio of dimensions. Put
\begin{equation}\label{eq:kappa}
\kappa_n:=\frac{D_n(\lp)}{D_n(\lambda)}\qquad(j\leq n\leq d),
\end{equation}
so that $\kappa_d=\dim V_{\lp}/\dim V_\lambda$ is the constant of \eqref{eq:schur}.

\begin{theorem}[Adding a box]\label{cor:adding}
For every density matrix $\sigma$ on $V_\lambda$,
\[
\Phi_+(\sigma)\mj\Phi_+(\sigma_0),\qquad \Psi_+(\sigma)\mj\Psi_+(\sigma_0).
\]
The coherent-state output $\Psi_+(\sigma_0)$ is diagonal in the basis $(e_a)$, with $d-j+1$ nonzero eigenvalues, and
\begin{equation}\label{eq:kappasums}
s_k\big(\Psi_+(\sigma_0)\big)=\frac{\kappa_{j+k-1}}{\kappa_d},\qquad 1\leq k\leq d-j+1 .
\end{equation}
\end{theorem}

The $k$ largest eigenvalues of $\kappa_d\,\Psi_+(\sigma_0)$ add up to $\kappa_{j+k-1}$, the ratio of the dimensions of the corresponding representations of $U(j+k-1)$. As for Theorem~\ref{thm:nested} below, it follows that the outputs of $\Psi_+$ are exactly the density matrices $A$ with $A\mj\Psi_+(\sigma_0)$. Applied to the dual pair of Lemma~\ref{lem:complement}, the hook-content formula turns \eqref{eq:kappasums} into the hook products of Theorem~\ref{thm:main} (Lemma~\ref{lem:translate}); for a single pair $(\lambda,\lp)$, the partial sums \eqref{eq:kappasums} of $\Psi_+$ and the hook products \eqref{eq:hooks} of $\Psi_-$ generally differ (Remark~\ref{rem:pm}).

\begin{remark}[Other covariant channels]\label{rem:six}
The isometries $T$ and $S$ are the same intertwiner with its factors arranged differently, and there is a third arrangement. Since $V_{\lp}$ occurs exactly once in $V_\lambda\otimes\Cd$, $\Cd$ occurs exactly once in $\overline{V_\lambda}\otimes V_{\lp}$, and $T$ becomes, up to the factor $(d/\dim V_{\lp})^{1/2}$, an isometric intertwiner $\tilde T:\Cd\to\overline{V_\lambda}\otimes V_{\lp}$ (compare~\cite[Prop.~18 and Rem.~13]{MT25}). Tracing out one factor or the other gives two more covariant channels, both with input $\Cd$. For these, optimality is trivial: $U(d)$ acts transitively on the pure states of $\Cd$, so all pure inputs give the same output spectrum, and $s_k$ is convex. Together with Theorems~\ref{thm:main} and~\ref{cor:adding}, this means that each of the six channels obtained in this way has a coherent input whose output majorizes all others.

For a general covariant channel between these representations, however, the coherent-state output need not majorize all others, and there may be no output that does. By~\cite[Thm.~19]{MT25}, a covariant channel from $\B(V_{\lp})$ to $\B(\Cd)$ has its environment inside $\Cdbar\otimes V_{\lp}=\bigoplus_rV_{\lp-e_r}$, and since each summand occurs once, the channel is a mixture $\sum_rq_r\Psi^{(r)}_-$ with one term for each summand. The summands belong to the removable boxes of $\lp$ and, if $\lp$ has fewer than $d$ rows, to $r=d$, where $\lp-e_d$ has the entry $-1$; up to a twist by $\det$, this last term is the channel for the removable box in row $d$ of $\lp+(1^d)$. The individual terms are covered by Theorem~\ref{thm:main}, but the mixtures are not, and they need not have an optimal state at all. Take $\lp=(3,2,1)$ and $d=3$. Two of its three removable boxes, $(2,2)$ and $(3,1)$, give coherent-state outputs with diagonals $(\tfrac13,\tfrac23,0)$ and $(\tfrac15,\tfrac4{15},\tfrac8{15})$ by \eqref{eq:entries}. For the mixture $\Lambda=\tfrac12\Psi^{(2,2)}_-+\tfrac12\Psi^{(3,1)}_-$, the coherent-state output has spectrum $(\tfrac7{15},\tfrac4{15},\tfrac4{15})$, while the unit vector proportional to $(E_{21}E_{32}+E_{32}E_{21})v_{\lp}$, of weight $(2,2,2)$, gives $(\tfrac25,\tfrac25,\tfrac15)$, and neither majorizes the other. In fact, no output majorizes all others: it would need $s_1\geq\tfrac7{15}$ and $s_2\geq\tfrac45$, whereas $\max_\rho\big[s_1(\Lambda(\rho))+s_2(\Lambda(\rho))\big]=1+\tfrac{2\sqrt3}{15}<\tfrac{19}{15}$. By Ky Fan's maximum principle and covariance, this maximum is the largest eigenvalue of $2\Lambda^*(|e_1\rangle\langle e_1|)+\Lambda^*(|e_2\rangle\langle e_2|)$, where $\Lambda^*$ is the adjoint of $\Lambda$; in a weight basis this operator is diagonal except on the weight space of weight $(2,2,2)$, where its eigenvalues are $1\pm\tfrac{2\sqrt3}{15}$. The obstruction can be seen in \eqref{eq:entries}: different boxes put their largest entry at different basis vectors, and no state can sit on both peaks.
\end{remark}

\begin{remark}[Cartan channels]\label{rem:cartan}
For $j=1$ the channel $\Phi_+$ is the Cartan channel \eqref{eq:cartan} with $V_\alpha=\Cd$, and for $j=d$ the channel $\Phi_-$ is the one with $V_\alpha=\Cdbar$: then $V_\lambda=V_{\lp-e_d}$ is the Cartan component of $V_{\lp}\otimes\Cdbar$, and $\Phi_-(\rho)=\kappa_dS^\dagger(\rho\otimes\one)S$ by \eqref{eq:S}. So Theorems~\ref{thm:main} and~\ref{cor:adding} answer the question of~\cite[eq.~(2.130)]{Aschieri25} affirmatively for these two $\alpha$ and every $\lambda$. For general $\alpha$ the answer is negative. For $d=3$, $\lambda=(8,4,0)$ and $\alpha=(4,2,0)$, the unit vector of weight $(6,6,0)$ in $V_\lambda$ beats the coherent state in the sum of the $26$ largest of the $27$ nonzero output eigenvalues of $\Lambda_{\lambda,\alpha}$, and the unit vector of weight $(1,1,4)$ in $V_\alpha$ beats it in the sum of the $124$ largest of the $125$ nonzero output eigenvalues of $\Lambda_{\alpha,\lambda}$. For $d=4$, $\lambda=(5,3,1,0)$ and $\alpha=(2,2,0,0)$, the unit vector proportional to $v_\lambda+\tfrac12E_{21}E_{43}v_\lambda$ beats it in the sum of the $6$ largest of the $20$ nonzero output eigenvalues. In these examples, however, coherent states still appear to minimize the output entropy.
\end{remark}

\begin{remark}[Where these channels occur]\label{rem:occur}
By~\cite[Thm.~19]{MT25}, the extreme points of the set of covariant channels from $V_\alpha$ to $V_\beta$ are indexed by the irreducible constituents of $\overline{V_\alpha}\otimes V_\beta$, together with a vector in the corresponding multiplicity space. Our $\Phi_+$ and $\Phi_-$ are the extreme points belonging to the constituents $\Cd$ and $\Cdbar$, whose multiplicity is one. They also occur inside the unitary-equivariant permutation-invariant channels on tensor powers of $\Cd$, which factor through covariant channels between irreducible representations~\cite{MT25}; for a single output copy, the extreme points at the level of irreducible representations are exactly the channels $\Psi^{(r)}_-$ of Remark~\ref{rem:six}~\cite{GrinkoOzols24,MT25}. Two examples: optimal purity amplification with a single output is, in each Schur--Weyl block, the channel $\Psi_-$ for a removable box~\cite{LFIC25},~\cite[eq.~(120)]{MT25} (optimal purification was studied earlier in~\cite{CEM99,KeylWerner01}; the single-output channels for other eigenvectors are again of this form~\cite{LTHC26}), and optimal symmetric cloning by one copy~\cite{Werner98} is $\Phi_+$ for a single row, whose complementary channel $\Psi_+$ is then optimal transposition into a single copy~\cite{BGSQ26}. The box chosen in purity amplification, in the topmost removable row, need not be the one that maximizes the bound of Corollary~\ref{thm:rdm} below, and there is no simple rule for which box does~\cite[Rem.~5]{Reuvers19}: for $\nu=(5,5,5,4)$ and $k=1$ the topmost box gives $\tfrac13$, the other one $\tfrac25$.
\end{remark}

\begin{remark}[The two one-particle outputs]\label{rem:pm}
The channels $\Psi_+$ and $\Psi_-$ both output a single particle, but they are different channels. For $\lambda=(2,1,0,0)$, $d=4$ and $j=2$, for instance, $\max_\sigma s_2(\Psi_+(\sigma))=\kappa_3/\kappa_4=\tfrac34$, whereas $\max_\rho s_2(\Psi_-(\rho))=1$ by \eqref{eq:hooks}. Only their largest output eigenvalues are tied together, $\max_\sigma s_1(\Psi_+(\sigma))=\kappa_d^{-1}\max_\rho s_1(\Psi_-(\rho))$, because $\|T_aT_a^\dagger\|=\|T_a^\dagger T_a\|$; Remark~\ref{rem:rectangle} exploits this. In the fermionic case $\lambda=(1^N)$, with the box added in row $N+1$, the duality of Lemma~\ref{lem:complement} is the exchange of particles and holes. The output of $\Psi_+$ then has the spectrum of $(\one-N\gamma)/(d-N)$, where $\gamma$ is the reduced density matrix of the input, and the Pauli principle for the $d-N$ holes is nothing but the positivity of $\gamma$.
\end{remark}

\subsection{Minimal output entropy and the minimizers}\label{sec:consequences}

Theorems~\ref{thm:main} and~\ref{cor:adding} give the entropy statements at once.

\begin{corollary}[Minimal output entropy]\label{cor:entropy}
For every concave $f$ and every density matrix $\rho$ on $V_{\lp}$,
\[
\Tr f\big(\Psi_-(\rho)\big)\ \geq\ \Tr f\big(\Psi_-(\rho_0)\big),
\]
and the same holds for $\Phi_-$, and for $\Phi_+$ and $\Psi_+$ with $\sigma,\sigma_0$ on $V_\lambda$ in place of $\rho,\rho_0$. In particular, coherent states minimize the von Neumann and all R\'enyi output entropies of the four channels.
\end{corollary}

All four channels are covariant with respect to irreducible representations on both sides, so their Holevo capacity is $\log(\dim\text{output})-H_{\min}$, with $H_{\min}$ the minimal output entropy~\cite{Holevo05}. Corollary~\ref{cor:entropy} evaluates it: $H_{\min}$ is the Shannon entropy of the coherent-state output spectrum, which is $p^{(j,m)}$ for the channels that remove the box, and can be read off from \eqref{eq:kappasums} for those that add it. The two channels of each pair have the same $H_{\min}$, but different output dimensions.

In general, however, coherent states are not the only minimizers. To describe all of them, let $G_j\subseteq U(d)$ be the subgroup fixing $e_1,\dots,e_j$, a copy of $U(d-j)$, and let $\mathcal S\subseteq V_{\lp}$ be the submodule that $G_j$ generates from $v_{\lp}$. It is an irreducible representation of $G_j$ with highest weight $(\lp_{j+1},\dots,\lp_d)$, the diagram $\lp$ with its first $j$ rows removed. In the Gelfand--Tsetlin basis of $V_{\lp}$, whose vectors are labelled by the semistandard tableaux of shape $\lp$ and have the content of their tableau as weight~\cite[Ch.~8]{GoodmanWallach}, the space $\mathcal S$ is spanned by the tableaux in which row $a$ is filled with the entry $a$ for every $a\leq j$, while the remaining rows form an arbitrary semistandard tableau with entries $j+1,\dots,d$. Equivalently, $\mathcal S$ is the sum of the weight spaces of $V_{\lp}$ whose weights agree with $\lp$ in the first $j$ coordinates. For comparison, $v_{\lp}$ is the tableau in which every row $a$ is filled with the entry $a$; so $\mathcal S$ keeps this filling in the first $j$ rows and frees the rows below.

\begin{theorem}[Minimizers]\label{thm:minimizers}
$\Psi_-(\rho)$ has spectrum $p^{(j,m)}$ if and only if $\rho=\pi_{\lp}(U)\,\sigma\,\pi_{\lp}(U)^{\dagger}$ for some $U\in U(d)$ and some density matrix $\sigma$ supported on $\mathcal S$. Every such $\sigma$ satisfies $\Psi_-(\sigma)=\Psi_-(\rho_0)$.
\end{theorem}

For strictly concave $f$, these are exactly the states that give equality in Corollary~\ref{cor:entropy}. The theorem also says something about the other three channels. Every pair $(\lambda,\lp)$ is the dual of another pair, and the proof of Lemma~\ref{lem:complement} identifies $\Psi_-$ with the channel $\Psi_+$ of the dual pair up to antiunitary conjugations; so Theorem~\ref{thm:minimizers} determines all minimizers of $\Psi_+$ as well. Since complementary channels have the same output spectra on pure inputs, it also determines the pure minimizers of $\Phi_-$ and $\Phi_+$. Their mixed minimizers are not covered here.

When are the minimizers coherent? In Section~\ref{sec:minimizers} we show that the coherent vectors of $V_{\lp}$ that lie in $\mathcal S$ are exactly the $G_j$-orbit of $v_{\lp}$, and deduce the following.

\begin{corollary}[Coherent minimizers]\label{cor:coherent}
\begin{enumerate}[label=\emph{(\alph*)},leftmargin=*]
\item Coherent states are the only minimizers exactly when $\mathcal S$ is one-dimensional, that is, when $(\lp_{j+1},\dots,\lp_d)$ is constant. This is the case, for example, when the box is removed from the last nonempty row of $\lp$.
\item Every pure minimizer is coherent exactly when $(\lp_{j+1},\dots,\lp_d)$ is constant, or of the form $(c+1,c,\dots,c)$ or $(c,\dots,c,c-1)$.
\end{enumerate}
\end{corollary}

Two examples. For $\lp=(2,2,1,0)$ and the box $(2,2)$ we have $\mathcal S\cong\C^2$, so the minimizers are the mixtures of coherent states in a conjugate of $\mathcal S$: all pure minimizers are coherent, but there are also mixed ones. In general, even the pure minimizers need not be coherent. Take $\lp=(3,2,0)$, $d=3$ and the box $(1,3)$. Here $\mathcal S$ is spanned by $v_{\lp}$, $E_{32}v_{\lp}$ and $E_{32}^2v_{\lp}$, and it is a copy of $\operatorname{Sym}^2\C^2$, whose coherent vectors are the squares. The unit vector $\phi=\frac1{\sqrt2}\big(v_{\lp}+\frac12E_{32}^2v_{\lp}\big)$ corresponds to $\frac1{\sqrt2}(e_2^2+e_3^2)$, which is not a square. So $\phi$ is not coherent, yet its output is the same as that of the coherent state, which is a pure state here since $j=1$.

\subsection{Reduced density matrices of \texorpdfstring{$N$}{N} particles}\label{sec:rdm}

Let $\nu$ be a Young diagram with $N$ boxes and at most $d$ rows, so that $V_\nu\otimes S_\nu\subset\otimes^N\Cd$, and write $\gamma:=\Tr_{1,\dots,N-1}\varrho$ for the reduced density matrix of the last particle of a state $\varrho$. Restricted to $S_{N-1}$, which permutes the first $N-1$ particles, $S_\nu$ splits into one summand for each removable box, as one sees by comparing Schur--Weyl duality for $N-1$ and for $N$ particles through Pieri's rule,
\begin{equation}\label{eq:branchSN}
S_\nu=\bigoplus_{(j,m)\ \mathrm{removable}}S_\nu^{(j,m)},\qquad S_\nu^{(j,m)}\cong S_{\nu-e_j}.
\end{equation}
The states on $V_\nu\otimes S_\nu^{(j,m)}$ are the states of symmetry $\nu$ whose first $N-1$ particles have symmetry $\nu-e_j$; in Young's orthogonal basis of $S_\nu$, the summand $S_\nu^{(j,m)}$ is spanned by the standard tableaux with the entry $N$ in the box $(j,m)$. For them, tracing out the first $N-1$ particles is the channel $\Psi_-$ for $\lp=\nu$ and the box $(j,m)$, applied to the $V_\nu$-marginal (Lemma~\ref{lem:nestedspace}; in terms of tableaux, this is~\cite[Cor.~7]{Reuvers19}).

\begin{theorem}[Given symmetry of the remaining particles]\label{thm:nested}
Let $(j,m)$ be a removable box of $\nu$, and let $\varrho$ be a density matrix on $V_\nu\otimes S_\nu^{(j,m)}$. Then $\spec\gamma\mj p^{(j,m)}$, that is,
\[
s_k(\gamma)\ \leq\ \prod_{i=1}^{j-k}\Big(1-\frac{1}{h_{(i,m)}}\Big)\qquad\text{for all }k,
\]
with the hook lengths of $\nu$. All these bounds are attained by $\varrho=|v_\nu\rangle\langle v_\nu|\otimes\sigma$, for any density matrix $\sigma$ on $S_\nu^{(j,m)}$, and every density matrix $\gamma$ with $\spec\gamma\mj p^{(j,m)}$ occurs.
\end{theorem}
The states that attain all bounds at once are those whose $V_\nu$-marginal is one of the states of Theorem~\ref{thm:minimizers}. For fermions, $\nu=(1^N)$ has the single removable box $(N,1)$, with $S_\nu^{(N,1)}=S_\nu$, and $h_{(i,1)}=N-i+1$. The product telescopes,
\begin{equation}\label{eq:pauli}
s_k(\gamma)\ \leq\ \prod_{i=1}^{N-k}\frac{N-i}{N-i+1}\ =\ \frac{k}{N},
\end{equation}
and Theorem~\ref{thm:nested} becomes the Pauli principle $\gamma\leq\one/N$, together with Coleman's theorem~\cite{Coleman63} that every such $\gamma$ occurs for mixed states. For fermions the bounds for $k\geq2$ follow from the one for $k=1$. In general they do not: for $\nu=(5,5,5,4)$ and the box $(4,4)$, the bound on $s_1$ is $\tfrac25$, but the bound on $s_2$ is $\tfrac35$, not $\tfrac45$.

The condition on the first $N-1$ particles matters. For a product state $v\otimes s$ whose vector $s\in S_\nu$ has components in several summands of \eqref{eq:branchSN}, $\gamma$ is the output of a mixture $\sum_bw_b\Psi^{(b)}_-$ of the channels for different boxes, and such a mixture need not have an optimal state; Remark~\ref{rem:six} gives an example for $\nu=(3,2,1)$ and $d=3$. What survives for a general state on $V_\nu\otimes S_\nu$ is a decomposition: the cross terms between different summands vanish, so that $\gamma$ is a convex combination of reduced density matrices of states of the kind in Theorem~\ref{thm:nested} (Section~\ref{sec:corproofs}). This gives the following.

\begin{corollary}[States of a given permutation symmetry]\label{thm:rdm}
For every density matrix $\varrho$ on $V_\nu\otimes S_\nu$ and every $k$,
\[
s_k(\gamma)\ \leq\ \max_{(j,m)\ \mathrm{removable}}\ \prod_{i=1}^{j-k}\Big(1-\frac{1}{h_{(i,m)}}\Big).
\]
For each $k$ the bound is attained by the states of Theorem~\ref{thm:nested} for a box that maximizes the right-hand side. Different $k$ may require different boxes.
\end{corollary}

For $k=1$ this is \eqref{eq:hook1}. As different $k$ may require different boxes, these bounds do not describe the set of reduced density matrices that occur. That set is the convex hull of the sets of Theorem~\ref{thm:nested}.

\begin{corollary}[Reachable reduced density matrices]\label{thm:polytope}
As $\varrho$ varies over the density matrices on $V_\nu\otimes S_\nu$, the $\gamma$ that occur are exactly those with
\begin{equation}\label{eq:polytope}
\spec\gamma\ \mj \sum_{(j,m)\ \mathrm{removable}}c_{(j,m)}\,p^{(j,m)}
\end{equation}
for some probability vector $c$. Their eigenvalue vectors, listed in any order, form the convex hull of all permutations of the vectors $p^{(j,m)}$; in particular, the vertices of this polytope are among these permutations.
\end{corollary}

The bounds of Corollary~\ref{thm:rdm} can cut out a strictly larger region. For $\nu=(5,5,5,4)$, the two removable boxes give $p^{(3,5)}=(\tfrac13,\tfrac13,\tfrac13,0)$ and $p^{(4,4)}=(\tfrac25,\tfrac15,\tfrac15,\tfrac15)$, so the bounds read $s_1\leq\tfrac25$, $s_2\leq\tfrac23$, $s_3\leq1$. The vector $(\tfrac25,\tfrac4{15},\tfrac4{15},\tfrac1{15})$ satisfies all three, but it is the spectrum of no $\gamma$: equality in $s_1$ forces $c$ onto the second box, where $s_2=\tfrac35$.

Pure states suffice as soon as there is room. For a pure state on $V_\nu\otimes S_\nu^{(j,m)}$, the $V_\nu$-marginal can be any density matrix of rank at most $f^{\nu-e_j}=\dim S_{\nu-e_j}$, the number of standard tableaux of shape $\nu-e_j$. By Carath\'eodory's theorem, applied to the permutohedron of $p^{(j,m)}$, which has dimension at most $d-1$, every density matrix $A$ with $\spec A\mj p^{(j,m)}$ is a mixture of at most $d$ unitary conjugates of $\Psi_-(\rho_0)$, and hence the output of a mixture of at most $d$ coherent states. So pure states reach every $\gamma$ of Theorem~\ref{thm:nested} if $f^{\nu-e_j}\geq d$, and, by the decomposition above, every $\gamma$ of Corollary~\ref{thm:polytope} if this holds for every removable box. For fermions $f^{\nu-e_N}=1$, and the $\gamma$ of pure states satisfy the further constraints of~\cite{AK08} (see Section~\ref{sec:open}).

\begin{remark}[Parastatistics, and the symmetrized density matrix]\label{rem:para}
States of permutation symmetry $\nu$ are also the sectors of Green's parastatistics~\cite{Green53,MG64}; para-Fermi statistics of order $p$ corresponds to diagrams with at most $p$ columns. For identical paraparticles, however, only permutation-invariant observables are physical, and the relevant one-particle density matrix is not $\gamma$ but the average $\gamma_{\mathrm{sym}}$ of the reduced density matrices of the $N$ particles. Its theory is simpler. Since $\Tr(X\,N\gamma_{\mathrm{sym}})=\Tr\big(\varrho\,(\pi_\nu(X)\otimes\one)\big)$ for $X\in\gl_d$, it depends only on the $V_\nu$-reduced density matrix of the state, and the moment-polytope approach of~\cite{AK08} applies to it. For mixed states, the spectra of $N\gamma_{\mathrm{sym}}$ are exactly the vectors majorized by $\nu$: the diagonal of $N\gamma_{\mathrm{sym}}$ in any basis is a convex combination of weights of $V_\nu$, and mixtures of coherent states reach every such vector. In particular, no orbital holds more than $\nu_1$ particles, which for at most $p$ columns is the exclusion principle usually stated for parafermions. Theorem~\ref{thm:nested} and Corollaries~\ref{thm:rdm} and~\ref{thm:polytope} instead concern a single, distinguished particle, as for $N$ distinguishable systems whose joint state has symmetry $\nu$. For fermions all particles have the same reduced density matrix, the two pictures agree, and both reduce to the Pauli principle and Coleman's theorem.
\end{remark}

\begin{remark}[Relation to~\cite{RGKZ26}]\label{rem:grinko}
By Lemma~\ref{lem:nestedspace} and the proof of Corollary~\ref{thm:rdm}, the operator $\one^{\otimes(N-1)}\otimes|e_1\rangle\langle e_1|$ compressed to $V_\nu\otimes S_\nu$, which appears in the introduction, is unitarily equivalent to the direct sum over the removable boxes $(j,m)$ of the operators $T_1^\dagger T_1$ for $\lp=\nu$ and that box, each tensored with the identity on $S_{\nu-e_j}$. For a given box, the largest eigenvalue of $T_a^\dagger T_a$ is the product in \eqref{eq:hook1}, so \eqref{eq:hook1} is the largest eigenvalue of the whole operator. The operator $T_a^\dagger T_a$ commutes with the subgroup fixing $e_a$, a copy of $U(d-1)$, under which $V_{\lp}$ is multiplicity free, and~\cite{RGKZ26} compute its whole spectrum (see their Theorem~1 and Eq.~(14)), from the classical formula for squares of reduced Wigner coefficients~\cite{VilenkinKlimyk}. For $k\geq2$ the relevant operator $\sum_{a\leq k}T_a^\dagger T_a$ is only invariant under the subgroup preserving $\operatorname{span}\{e_1,\dots,e_k\}$, a copy of $U(k)\times U(d-k)$, under which $V_{\lp}$ has multiplicities, so this approach does not extend.
\end{remark}

\subsection{Open problems}\label{sec:open}

\begin{enumerate}[leftmargin=*]
\item \emph{Mixed minimizers.} Theorem~\ref{thm:minimizers} finds all minimizers of $\Psi_\pm$, but only the pure minimizers of $\Phi_\pm$, because complementary channels share their output spectra only on pure inputs. What are the mixed minimizers of $\Phi_\pm$?

\item \emph{Pure states.} If $f^{\nu-e_j}\geq d$ for every removable box, the pure states in $V_\nu\otimes S_\nu$ reach the whole polytope of Corollary~\ref{thm:polytope}, and for $\nu=(1^N)$ they satisfy the further constraints of~\cite{AK08}. For which $\nu$ in between do pure states reach the whole polytope, and which set of spectra do they reach when they do not? For fermions this set is a convex polytope, because $N\gamma$ is then a moment map for the action of $U(d)$ on $\wedge^N\Cd$~\cite{AK08}. This uses that all particles have the same reduced density matrix; for general $\nu$ the argument applies to $\gamma_{\mathrm{sym}}$ of Remark~\ref{rem:para}, but not to the reduced density matrix $\gamma$ of a single particle, and there is no reason to expect a polytope. Numerically, in some cases the set indeed appears not to be one.

\item \emph{Other groups.} Let $G = \mathrm{Sp}(2r)$ or $G = \mathrm{SO}(n)$, with $W$ the defining representation, and let $V_\mu \subseteq V_\nu \otimes W$ with $\mu = \nu + e_i$ (for $\mathrm{SO}(2r)$ with $\nu_r \geq 0$). The two channels out of $V_\mu$ are restrictions of the corresponding $U(d)$ channels, so Theorem~\ref{thm:main} holds for them, with the spectra~\eqref{eq:hooks}. For the two channels out of $V_\nu$, coherent states maximize the largest output eigenvalue, by the argument of Remark~\ref{rem:rectangle} and because $T_i v_\mu$ is a multiple of $v_\nu$ (Lemma~\ref{lem:hw}). We conjecture that for $G = \mathrm{Sp}(2r)$ their outputs also majorize all other outputs; numerically we found no violation in the cases tested, up to $\mu_1 = 12$ for $\mathrm{Sp}(4)$, $\mu_1 = 6$ for $\mathrm{Sp}(6)$ and $\mu_1 = 3$ for $\mathrm{Sp}(8)$. For $\mathrm{SO}(n)$ numerical examples show that this can fail, for instance for $\mathrm{SO}(5)$ with $\nu = (7,1)$, $\mu = (7,2)$, and the same holds for $\mathrm{Spin}(n)$, for instance for $\mathrm{Spin}(6)$ with $\nu = (\tfrac52, \tfrac32, \tfrac12)$, $\mu = (\tfrac52, \tfrac52, \tfrac12)$.

\item \emph{Other tensor factors.} If $\Cd$ is replaced by $W=\wedge^h\Cd$ or $W=\operatorname{Sym}^h\Cd$, then $V_\lambda\otimes W$ is still multiplicity free, so the four channels are still defined for each of its summands. But the coherent-state output no longer always majorizes all others. For $d=5$, $W=\wedge^2\C^5$, $\lambda=(1,1,1,0,0)$ and the summand $\lp=(2,1,1,1,0)$, it fails for the two channels out of $V_{\lp}$, already for the largest output eigenvalue: the coherent state gives $\tfrac13$, while another state reaches $\tfrac25$. For which summands does the coherent-state output majorize all others?
\end{enumerate}

\section{Proofs}\label{sec:proof}

\subsection{Strategy}\label{sec:strategy}

Everything reduces to one family of operators on $V_\lambda$, built from the projection $P=TT^\dagger$ of $V_\lambda\otimes\Cd$ onto its summand $V_{\lp}$ (Section~\ref{sec:channels}):
\begin{equation}\label{eq:M}
M_{ab}:=T_aT_b^\dagger=(\one\otimes\langle e_a|)P(\one\otimes|e_b\rangle),\qquad M_I:=\sum_{i\in I}M_{ii}\quad(I\subseteq\{1,\dots,d\}).
\end{equation}
The diagonal terms $M_{ii}$ are positive, so $M_I$ is positive as well and $\|M_I\|$ is its largest eigenvalue, and $\sum_aM_{aa}=\kappa_d\one$ by \eqref{eq:schur}. The proof of Theorems~\ref{thm:main} and~\ref{cor:adding} has the following steps.
\begin{enumerate}[leftmargin=*]
\item \emph{From removing to adding.} By duality, the channels $\Phi_-$ and $\Psi_-$ have the same output spectra as the channels $\Phi_+'$ and $\Psi_+'$ of the dual pair, with coherent states corresponding to coherent states (Lemma~\ref{lem:complement}). So Theorem~\ref{thm:main} is Theorem~\ref{cor:adding} for the dual pair, once the dimension ratios are rewritten in terms of hook lengths (Lemma~\ref{lem:translate}). 
\item \emph{From majorization to an operator norm.} By covariance, the largest value of $s_k(\Psi_+(\sigma))$ is $\kappa_d^{-1}\|M_I\|$ for any $I$ with $|I|=k$, and the coherent-state output is diagonal (Lemma~\ref{lem:reduction}). We therefore need an upper bound on $\|M_I\|$, and the diagonal entries of the coherent-state output, which show that the bound is attained.
\item \emph{Two reductions to a special shape.} The columns to the right of the box are spectators: they can be factored out of $M_{ab}$ exactly (Proposition~\ref{prop:short}). The columns of height at least $j+k$ reach below the rows $j,\dots,j+k-1$ that carry the $k$ largest eigenvalues of the coherent-state output, and deleting them can only increase $\|M_I\|$ (Proposition~\ref{prop:tall}). What remains has the shape \eqref{eq:shape}: its rows down to the box all have the same length, and it has at most $j+k-1$ rows. See Figure~\ref{fig:reductions}.
\item \emph{The upper bound for the special shape.} The operator $\kappa_d^{-1}M_I$ has the same nonzero eigenvalues as $(\one\otimes\Pi_I)Q(\one\otimes\Pi_I)$, with $\Pi_I$ the projection onto $\operatorname{span}\{e_i:i\in I\}$. This is the compression of the projection $Q=SS^\dagger$, for the isometry $S$ of \eqref{eq:S}, onto an isotypic component of $V_{\lp}\otimes\Cdbar$. When the rows down to the box all have length $m$, this projection is dominated by a Casimir-type operator (Lemma~\ref{lem:major}), and the compressions of that operator are controlled by branching to $U(k)$ (Lemma~\ref{lem:compress}). This gives Proposition~\ref{prop:casimir}; the bound is sharp when, in addition, the diagram has at most $j+k-1$ rows. Combined with the reductions of step 3, this bounds $\|M_I\|$ for every diagram (Section~\ref{sec:assembly}).
\item \emph{The coherent state.} Let $n=j+k-1$. Restricted to the $U(n)$-submodule of $V_\lambda$ generated by $v_\lambda$, the projection $P$ becomes the projection onto a single irreducible representation of $U(n)$, and Schur's lemma turns $\sum_{a\leq n}M_{aa}$ into the ratio of dimensions $\kappa_n$ (Proposition~\ref{prop:coherent}). Weyl's dimension formula shows that these ratios match the bound of step 4 (Lemma~\ref{lem:weyl}).
\end{enumerate}
Section~\ref{sec:corproofs} then deduces Corollary~\ref{cor:entropy}, Theorem~\ref{thm:nested} and Corollaries~\ref{thm:rdm} and~\ref{thm:polytope}, and Section~\ref{sec:minimizers} determines the minimizers: optimality forces one eigenvector equation for each of the rows $1,\dots,j$, and these equations remove the rows one at a time until only $\mathcal S$ is left.

\begin{figure}[t]
\centering
\begin{tikzpicture}[scale=0.62]
\fill[red!15] (0,0) rectangle (1,-6);
\fill[blue!12] (5,0) rectangle (6,-2);
\foreach \r/\len in {1/6,2/6,3/5,4/5,5/2,6/1}{
  \foreach \c in {1,...,\len}{ \draw (\c-1,-\r+1) rectangle (\c,-\r); }
}
\node at (4.5,-3.5) {$\bullet$};
\draw[decorate,decoration={brace,amplitude=4pt}] (6.5,-3) -- (6.5,-5) node[midway,right=6pt] {\small rows $j,\dots,j+k-1$};
\end{tikzpicture}
\caption{The two reductions of step 3, for $\lp=(6,6,5,5,2,1)$, the added box $(4,5)$ (marked $\bullet$) and $k=2$. The column to the right of the box (blue) can be deleted exactly (Proposition~\ref{prop:short}). The first column (red) has height $6\geq j+k$, so it reaches below the rows $j,\dots,j+k-1$; deleting it can only increase $\|M_I\|$ (Proposition~\ref{prop:tall}). What remains, $(4,4,4,4,1)$, has the shape \eqref{eq:shape} for which the Casimir bound is sharp.}
\label{fig:reductions}
\end{figure}
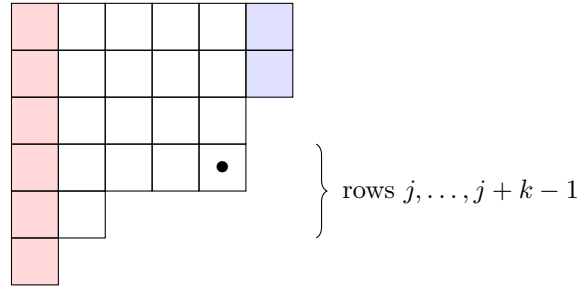

\subsection{Reduction to an operator norm}\label{sec:setup}

From \eqref{eq:S}, for a unit vector $\phi\in V_\lambda$,
\begin{equation}\label{eq:gram}
\kappa_d\,\langle e_a|\Psi_+(|\phi\rangle\langle\phi|)|e_b\rangle=\ip{T_b^\dagger\phi}{T_a^\dagger\phi}=\ip{\phi}{M_{ba}\phi}=\ip{P(\phi\otimes e_b)}{P(\phi\otimes e_a)} .
\end{equation}

\begin{lemma}\label{lem:reduction}
For every $I$ with $|I|=k$,
\begin{equation}\label{eq:maxsk}
\max_\sigma\,s_k\big(\Psi_+(\sigma)\big)=\kappa_d^{-1}\|M_I\|,
\end{equation}
where the maximum runs over all density matrices on $V_\lambda$. The coherent-state output $\Psi_+(\sigma_0)$ is diagonal in the basis $(e_a)$, with entries $g_a:=\kappa_d^{-1}\ip{v_\lambda}{M_{aa}v_\lambda}$.
\end{lemma}

\begin{proof}
It suffices to consider pure states, because $s_k$ is convex. For a Hermitian $A$ on $\Cd$, $s_k(A)=\max_\Pi\Tr(\Pi A)$, the maximum running over the rank-$k$ orthogonal projections $\Pi$, by Ky Fan's maximum principle~\cite[Ch.~20]{MOA}. By \eqref{eq:gram},
\[
\Tr\big(\Pi\,\Psi_+(|\phi\rangle\langle\phi|)\big)=\kappa_d^{-1}\ip{\phi}{M_\Pi\phi},\qquad M_\Pi:=\sum_{a,b}\Pi_{ab}M_{ab},\quad \Pi_{ab}:=\langle e_a|\Pi|e_b\rangle .
\]
Covariance of $T$ reads $T_a\pi_{\lp}(U)=\pi_\lambda(U)\sum_bU_{ab}T_b$, and this gives $\pi_\lambda(U)^\dagger M_\Pi\,\pi_\lambda(U)=M_{\Pi'}$ with $\Pi'=\overline U^\dagger\Pi\,\overline U$, the conjugate action inherited from $\Cdbar$. As $\overline U$ runs over $U(d)$, $\Pi'$ runs over all projections of the same rank, so $\|M_\Pi\|$ depends only on the rank of $\Pi$ and equals $\|M_I\|$. Finally, the coherent-state output is diagonal because the vectors $T_a^\dagger v_\lambda$ have the distinct weights $\lambda+e_a$.
\end{proof}

So Theorem~\ref{cor:adding} follows from two statements: an upper bound on $\|M_I\|$,
\begin{equation}\label{eq:target}
\|M_I\|\leq\kappa_{j+k-1}\qquad\text{for }|I|=k\leq d-j+1,
\end{equation}
and the values of the $g_a$,
\begin{equation}\label{eq:target2}
g_a=0\ \text{ for }a<j,\qquad g_a>0\ \text{ for }a\geq j,\qquad g_j+\dots+g_{j+k-1}=\kappa_{j+k-1}/\kappa_d .
\end{equation}
The first is proved in Sections~\ref{sec:reduction} and~\ref{sec:casimir} and assembled in Section~\ref{sec:assembly}, the second in Section~\ref{sec:coherent}.

\subsection{Highest weight vectors and dimension ratios}\label{sec:facts}

When $T$, $P$, $M_{ab}$ and $\kappa_n$ are formed for a diagram $\beta$ other than $\lambda$, always with the box added in the same row $j$ (so that $\beta+e_j$ must be a Young diagram), we write $\beta$ as a superscript or argument: $T^\beta$, $P^\beta$, $M^\beta_{ab}$, $\kappa_n(\beta)$. We also recall that $V_{\alpha+\beta}$ occurs exactly once in $V_\alpha\otimes V_\beta$ and is generated by $v_\alpha\otimes v_\beta$, since $\alpha+\beta$ is the highest weight of the tensor product. We call it the \emph{Cartan component} and write $V_{\alpha+\beta}\subseteq V_\alpha\otimes V_\beta$.

Both \eqref{eq:target} and \eqref{eq:target2} use the subgroup $U(n)\subseteq U(d)$, through the $\gl_n$-submodule $L_n(\mu)\subseteq V_\mu$ generated by $v_\mu$. This submodule is irreducible, with highest weight $(\mu_1,\dots,\mu_n)$ and dimension $D_n(\mu)$. The following facts will be used throughout.

\begin{lemma}\label{lem:index}
Every weight vector of $V_\mu$ whose weight agrees with $\mu$ in the coordinates $n+1,\dots,d$ lies in $L_n(\mu)$.
\end{lemma}
\begin{proof}
Such a vector is a linear combination of monomials $E_{a_1b_1}\cdots E_{a_rb_r}v_\mu$ with $a_s>b_s$, each of the given weight. Each factor $E_{ab}$ moves one unit of weight from the coordinate $b$ to the later coordinate $a$. So for every $t$, the sum of the coordinates $t,\dots,d$ of the weight can only grow along the monomial, and it grows at each factor with $b_s<t\leq a_s$. For $t>n$ these sums end where they started, since the weight agrees with $\mu$ beyond the $n$-th coordinate. Taking $t=a_s$ shows that no factor has $a_s>n$. So all indices are at most $n$, and the monomial lies in $L_n(\mu)$.
\end{proof}

\begin{lemma}\label{lem:branch}
Every irreducible constituent of $V_\mu|_{U(n)}$ has a highest weight $\nu$ with $\nu_a\leq\mu_a$ for $a\leq n$.
\end{lemma}

\begin{proof}
By the Gelfand--Tsetlin branching rule, the constituents of $V_\mu|_{U(d-1)}$ are the $\nu$ that interlace $\mu$ \cite[\S8.1]{GoodmanWallach}. Now iterate.
\end{proof}

\begin{lemma}\label{lem:hw}
Let $x_a:=T_av_{\lp}$, so that $Tv_{\lp}=\sum_ax_a\otimes e_a$. Then $x_a$ has weight $\lp-e_a$, it vanishes for $a>j$, and it lies in $L_j(\lambda)$ for $a\leq j$.
\end{lemma}

\begin{proof}
$Tv_{\lp}$ has weight $\lp$, so $x_a$ has weight $\lp-e_a=\lambda+e_j-e_a$. For $a>j$ this weight exceeds $\lambda$ in the dominance order, so $x_a=0$. For $a\leq j$ it agrees with $\lambda$ beyond the $j$-th coordinate, and Lemma~\ref{lem:index} applies.
\end{proof}

\begin{lemma}[Weyl ratios]\label{lem:weyl}
With $\omega_a:=\lambda_a-a$ and $j\leq n\leq d$,
\begin{equation}\label{eq:weyl}
\kappa_n=\prod_{a<j}\frac{\omega_a-\omega_j-1}{\omega_a-\omega_j}\prod_{j<b\leq n}\frac{\omega_j-\omega_b+1}{\omega_j-\omega_b},\qquad\text{so that}\qquad
\frac{\kappa_n}{\kappa_d}=\prod_{n<b\leq d}\frac{\omega_j-\omega_b}{\omega_j-\omega_b+1}.
\end{equation}
In particular: \emph{(a)} $\kappa_n/\kappa_d$ depends only on $\lambda_j,\dots,\lambda_d$; \emph{(b)} $\kappa_n$ does not change if a column of height $\geq n$ is deleted from $\lambda$; \emph{(c)} if $\lambda_b=0$ for $b>n$, then $\kappa_n/\kappa_d=\frac{m+n-j}{m+d-j}$, and otherwise $\kappa_n/\kappa_d$ is strictly smaller.
\end{lemma}

\begin{proof}
Weyl's dimension formula $D_n(\mu)=\prod_{a<b\leq n}\frac{(\mu_a-a)-(\mu_b-b)}{b-a}$~\cite[\S7.1.2]{GoodmanWallach} gives \eqref{eq:weyl}, because adding the box raises $\omega_j$ by one and changes only the factors that involve $j$. Then (a) is immediate, and (b) holds because deleting such a column lowers $\omega_1,\dots,\omega_n$ by the same amount. For (c), $\omega_j-\omega_b=m-1-j+b$, and the product telescopes. If $\lambda_b>0$ for some $b>n$, the corresponding factor $x/(x+1)$, with $x=\omega_j-\omega_b$, is smaller, since $x$ is.
\end{proof}

\subsection{Deleting columns}\label{sec:reduction}

The bound \eqref{eq:target} will be proved in Section~\ref{sec:casimir} under two hypotheses on the shape of the diagram,
\begin{equation}\label{eq:shape}
\lp_1=\dots=\lp_j\ (=m),\qquad \lp_{j+k}=0 .
\end{equation}
That is, all rows down to the added box have length $m$, and the diagram has at most $j+k-1$ rows. In this section we remove the columns that violate these hypotheses: those of height $<j$, which lie to the right of the box, and those of height $\geq j+k$, which lie to its left (Figure~\ref{fig:reductions}).

\begin{proposition}[Columns shorter than the box]\label{prop:short}
Let $h<j$, let $\hat\lambda$ be a Young diagram such that $\hat\lambda+e_j$ is also a Young diagram, and let $\lambda=\hat\lambda+(1^h)$. Let $\iota:V_\lambda\to V_{\hat\lambda}\otimes\wedge^h\Cd$ be the Cartan embedding. Then
\begin{equation}\label{eq:short}
\frac{M^\lambda_{ab}}{\kappa_d(\lambda)}=\iota^\dagger\Big(\frac{M^{\hat\lambda}_{ab}}{\kappa_d(\hat\lambda)}\otimes\one\Big)\iota\qquad\text{for all }a,b .
\end{equation}
In particular, $\|M^\lambda_I\|/\kappa_d(\lambda)\leq\|M^{\hat\lambda}_I\|/\kappa_d(\hat\lambda)$ for every $I$.
\end{proposition}

In words: the column of height $h$ is carried along as a spectator, and the operators $M_{ab}$ do not notice it.

\begin{proof}
Let $\iota^+:V_{\lp}\to V_{\hat\lambda+e_j}\otimes\wedge^h\Cd$ be the Cartan embedding, and regard $T^{\hat\lambda}\otimes\one$ as a map into $V_{\hat\lambda}\otimes\wedge^h\Cd\otimes\Cd$ by moving the last two factors past each other. Then
\[
\iota^\dagger\,(T^{\hat\lambda}\otimes\one)\,\iota^+:V_{\lp}\to V_\lambda\otimes\Cd
\]
is an intertwiner, hence equal to $\theta T$ for some scalar $\theta$, and its $a$-th component is $\iota^\dagger(T^{\hat\lambda}_a\otimes\one)\iota^+$.

Multiplying two such components gives $|\theta|^2M^\lambda_{ab}$, and we claim that the factor $\iota^+\iota^{+\dagger}$ in the middle may be dropped. To see this, consider the adjoint construction $\phi\mapsto(T^{\hat\lambda\dagger}_b\otimes\one)\iota\phi$. It is the $b$-th component of an intertwiner $V_\lambda\otimes\Cd\to V_{\hat\lambda+e_j}\otimes\wedge^h\Cd$. The source has the constituents $V_{\lambda+e_l}$, and by Pieri's rule for $\wedge^h$ the target has the constituents $V_{\hat\lambda+e_j+e_J}$ with $e_J:=\sum_{i\in J}e_i$ and $|J|=h$. A common constituent requires $(1^h)+e_l=e_j+e_J$, and since $j>h$ this forces $l=j$ and $J=\{1,\dots,h\}$. So the intertwiner kills $V_{\lambda+e_l}$ for $l\neq j$ and maps $V_{\lp}$ into the Cartan component, which is the range of $\iota^+$. Hence $|\theta|^2M^\lambda_{ab}=\iota^\dagger(M^{\hat\lambda}_{ab}\otimes\one)\iota$. Summing over $a=b$ gives $|\theta|^2\kappa_d(\lambda)=\kappa_d(\hat\lambda)$ by \eqref{eq:schur}, and dividing the two identities gives \eqref{eq:short}.
\end{proof}

\begin{corollary}\label{cor:cap}
Let $\tilde\lambda_a:=\min(\lambda_a,m)$, the diagram $\lambda$ with the columns to the right of column $m$ deleted. Then $\tilde\lambda$ agrees with $\lambda$ in the rows $a\geq j$, the diagram $\tilde\lambda+e_j$ satisfies the first hypothesis in \eqref{eq:shape}, and
\[
\frac{\|M^\lambda_I\|}{\kappa_d(\lambda)}\leq\frac{\|M^{\tilde\lambda}_I\|}{\kappa_d(\tilde\lambda)}\qquad\text{for every }I.
\]
\end{corollary}

\begin{proof}
The deleted columns have height at most $j-1$, since $\lambda_j=m-1$. Delete them one at a time from the right, applying Proposition~\ref{prop:short} at each step. The hypothesis of the proposition is inherited: no deletion touches row $j$, and none brings row $j-1$ below $m$, so $\hat\lambda_{j-1}\geq m=\hat\lambda_j+1$ throughout.
\end{proof}

\begin{proposition}[Columns at least as tall as the box]\label{prop:tall}
Let $\alpha,\beta$ be Young diagrams with $\alpha_1=\dots=\alpha_j$, that is, with every column of $\alpha$ of height at least $j$, and with $\beta+e_j$ a Young diagram. Let $\lambda=\alpha+\beta$. Then $\|M^\lambda_I\|\leq\|M^\beta_I\|$ for every $I$.
\end{proposition}

\begin{proof}
Inside $V_\alpha\otimes V_\beta\otimes\Cd$ we have $V_\lambda\otimes\Cd$, through the Cartan embedding, and $V_\alpha\otimes V_{\beta+e_j}$, through $\one_\alpha\otimes T^\beta$. We first show that the range $V_{\lp}$ of $P^\lambda$ lies in the second:
\begin{equation}\label{eq:incl}
V_{\lp}\subseteq V_\alpha\otimes V_{\beta+e_j}.
\end{equation}
The vector $u:=v_\alpha\otimes T^\beta v_{\beta+e_j}$ in $V_\alpha\otimes V_{\beta+e_j}$ is a highest weight vector of weight $\lp$. It therefore suffices to show $u\in V_\lambda\otimes\Cd$, because $V_{\lp}$ occurs only once there, so $u$ generates it. By Lemma~\ref{lem:hw}, $u=\sum_{a\leq j}v_\alpha\otimes x_a\otimes e_a$ with $x_a\in L_j(\beta)$. Since $\alpha_1=\dots=\alpha_j$, the $E_{ab}$ with $a\neq b\leq j$ annihilate $v_\alpha$: for $a<b$ because $v_\alpha$ is a highest weight vector, and for $a>b$ because $\|E_{ab}v_\alpha\|^2=\ip{v_\alpha}{[E_{ba},E_{ab}]v_\alpha}=\alpha_b-\alpha_a=0$. So $\gl_j$ acts on $v_\alpha$ by scalars, and $v_\alpha\otimes L_j(\beta)=L_j(\lambda)\subseteq V_\lambda$.

Extending $P^\lambda$ by zero, \eqref{eq:incl} says $P^\lambda\leq\one_\alpha\otimes P^\beta$. Hence, for a unit vector $\phi\in V_\lambda$ with $\rho_\beta:=\Tr_{V_\alpha}|\phi\rangle\langle\phi|$,
\[
\ip{\phi}{M^\lambda_I\phi}=\sum_{i\in I}\|P^\lambda(\phi\otimes e_i)\|^2\leq\sum_{i\in I}\ip{\phi}{(\one_\alpha\otimes M^\beta_{ii})\phi}=\Tr\big(\rho_\beta M^\beta_I\big)\leq\|M^\beta_I\|.\qedhere
\]
\end{proof}

\begin{remark}[The two reductions]\label{rem:twokinds}
In the Schur--Weyl language of~\cite{Reuvers19}, the inclusion \eqref{eq:incl} says that the range of the projection $P_t$ of Young's orthogonal basis for a column-ordered tableau $t$ lies in the tensor product of the ranges for its left and right parts~\cite[Cor.~8, eq.~(36)]{Reuvers19}, and for $k=1$ the reduction itself is~\cite[eq.~(55)]{Reuvers19}. The two reductions are of different kinds. Proposition~\ref{prop:short} is an identity: a short column is carried along as an inert tensor factor, and this is what keeps the coherent state optimal afterwards. Proposition~\ref{prop:tall} only gives an inequality, and deleting a tall column does change the channel. Neither argument covers the other case: \eqref{eq:incl} needs $\alpha_1=\dots=\alpha_j$ and fails for a short column, while the rigidity in Proposition~\ref{prop:short} comes from $j>h$, which forces a unique solution of $(1^h)+e_l=e_j+e_J$.
\end{remark}

\subsection{The Casimir bound}\label{sec:casimir}

In this section $\lp_1=\dots=\lp_j=m$. Let $Q:=SS^\dagger$ be the projection of $V_{\lp}\otimes\Cdbar$ onto its summand of type $\lambda$, where $S$ is the isometry \eqref{eq:S}. This summand is unique because $V_{\lp}\otimes\Cdbar\cong\bigoplus_rV_{\lp-e_r}$, the sum running over the $r$ for which $\lp-e_r$ is dominant, is multiplicity free. Let $\Pi_I$ be the projection onto $\operatorname{span}\{e_i:i\in I\}$. By \eqref{eq:gram}, $M_I=\kappa_dS^\dagger(\one\otimes\Pi_I)S$, and since $\|X^\dagger X\|=\|XX^\dagger\|$,
\begin{equation}\label{eq:flip}
\|M_I\|=\kappa_d\,\big\|(\one\otimes\Pi_I)\,Q\,(\one\otimes\Pi_I)\big\|.
\end{equation}
This is what the isometry $S$ is for. The operator $M_I$ is not a compression of a projection, but by \eqref{eq:flip} it has the same norm as $(\one\otimes\Pi_I)Q(\one\otimes\Pi_I)$, which is, and projections onto isotypic components are dominated by Casimir operators.

\begin{lemma}\label{lem:casimir}
For a dominant integral weight $\mu$, let $\Theta_\mu:=\sum_{a,b}\pi_\mu(E_{ab})\otimes E_{ab}$ on $V_\mu\otimes\Cdbar$, where $E_{ab}$ in the second factor is the matrix unit. Then $\Theta_\mu$ acts on the summand $V_{\mu-e_r}$ as the scalar $\mu_r+d-r$.
\end{lemma}

\begin{proof}
On $\Cdbar$ the generator $E_{ab}$ is represented by $-E_{ba}$, so $\Theta_\mu=-\tfrac12(C-C_\mu-C_{\bar\square})$, where $C$ is the quadratic Casimir $\sum_{a,b}E_{ab}E_{ba}$ of the tensor product and $C_\mu$, $C_{\bar\square}$ are those of the two factors. The Casimir acts on $V_\nu$ as $\sum_a\nu_a(\nu_a+d+1-2a)$, as one sees on the highest weight vector: for $a<b$, $E_{ab}E_{ba}v_\nu=(\nu_a-\nu_b)v_\nu$ and $E_{ba}E_{ab}v_\nu=0$. Taking $\nu$ to be $\mu-e_r$, $\mu$ and $(0,\dots,0,-1)$ in turn gives $-\tfrac12\big[(-2\mu_r-d+2r)-d\big]=\mu_r+d-r$.
\end{proof}

\begin{lemma}\label{lem:major}
$Q\leq\Theta_{\lp}/(m+d-j)$.
\end{lemma}

\begin{proof}
No row above $j$ is removable from $\lp$, so the summands that occur are the $V_{\lp-e_r}$ with $r\geq j$. By Lemma~\ref{lem:casimir}, $\Theta_{\lp}$ equals $m+d-j$ on $V_{\lp-e_j}$, which is the range of $Q$, and $\lp_r+d-r\geq0$ on the others.
\end{proof}

Note that only the nonnegativity of the other eigenvalues is used: there is no gap and no counting of multiplicities.

\begin{lemma}\label{lem:compress}
For every $\mu$ and every $I$ with $|I|=k$,
$(\one\otimes\Pi_I)\,\Theta_\mu\,(\one\otimes\Pi_I)\leq(\mu_1+k-1)\,(\one\otimes\Pi_I)$.
\end{lemma}

\begin{proof}
Let $U(I)\cong U(k)$ be the subgroup acting on the range of $\Pi_I$. The compression is $\sum_{a,b\in I}\pi_\mu(E_{ab})\otimes E_{ab}$ on $V_\mu\otimes\operatorname{ran}\Pi_I$, which is the operator of Lemma~\ref{lem:casimir} for the group $U(I)$. So if $V_\nu$ is a constituent of $V_\mu|_{U(I)}$, that lemma gives the eigenvalues $\nu_s+k-s\leq\nu_1+k-1$ on the corresponding summands, and $\nu_1\leq\mu_1$ by Lemma~\ref{lem:branch}, after permuting the coordinates so that $I=\{1,\dots,k\}$.
\end{proof}

\begin{proposition}\label{prop:casimir}
If $\lp_1=\dots=\lp_j=m$ then, for every $I$ with $|I|=k$,
\[
\|M_I\|\ \leq\ \kappa_d\,\frac{m+k-1}{m+d-j}.
\]
\end{proposition}

\begin{proof}
Combine \eqref{eq:flip} with Lemmas~\ref{lem:major} and~\ref{lem:compress}, the latter with $\mu=\lp$ and $\mu_1=m$.
\end{proof}

By Lemma~\ref{lem:weyl}(c), the bound equals $\kappa_{j+k-1}$ exactly when $\lp_{j+k}=0$, which is the second hypothesis in \eqref{eq:shape}. Otherwise it is strictly larger, and this is what Proposition~\ref{prop:tall} repairs. Here is an example.

\begin{example}\label{ex:211}
Take $d=3$, $\lambda=(2,1,1)$ and the box $(2,2)$, so that $\lp=(2,2,1)$. Up to a twist by $\det$, this is the passage from $\C^3$ to $\wedge^2\C^3$, that is, two fermions, and the true value is $\|M_{aa}\|=\kappa_2=\tfrac12$, the Pauli principle. Since $\lp$ has three rows, Proposition~\ref{prop:casimir} only gives $\tfrac23$, with $\kappa_3=1$. Deleting the column of height $3$ (Proposition~\ref{prop:tall}) leaves $\beta=(1,0,0)$, for which the bound is sharp.
\end{example}

\begin{remark}[The case $k=1$]\label{rem:rectangle}
For $k=1$ the Casimir estimate can be avoided. Deleting the columns to the left of the box (Proposition~\ref{prop:tall}) moves the box into the first column and leaves $j$ rows, without changing the hook lengths in the column of the box. By Lemma~\ref{lem:complement}, $\Psi_-$ then has the output spectra of the adding channel of the dual pair. At $k=1$, the adding and removing channels of a pair have the same largest output eigenvalue up to the factor $\kappa_d$, since $\|T_a\|=\|T_a^\dagger\|$. In the dual pair the box lies in the last column, of height $d+1-j$, and deleting the columns to its left leaves a single column, where the bound is the Pauli principle. Since $\dim V_{\mu^c}=\dim V_\mu$, this gives $\max_\rho s_1\big(\Psi_-(\rho)\big)\leq\kappa_d/(d+1-j)$ for the diagram left after the first deletion, and this equals $\kappa_j$ by Lemma~\ref{lem:weyl}(c). The first deletion does not change $\kappa_j$ (Lemma~\ref{lem:weyl}(b)), so the bound $\kappa_j$ holds for the original diagram, and Proposition~\ref{prop:coherent} with $n=j$ gives $\ip{v_\lambda}{M_{jj}v_\lambda}=\kappa_j$. The coherent state of $V_{\lp}$ attains it too, since $T_jv_{\lp}$ and $T_j^\dagger v_\lambda$ are multiples of $v_\lambda$ and $v_{\lp}$ of the same modulus. The hook-content formula gives $\kappa_j=\prod_{i<j}(1-1/h_{(i,m)})$. This is a short proof of \eqref{eq:hook1}, which was obtained in~\cite{Reuvers19} by a considerably longer argument.
\end{remark}

\subsection{The coherent state}\label{sec:coherent}

For the lower bound we use the coherent state as a trial state, and compute the diagonal entries $\ip{v_\lambda}{M_{aa}v_\lambda}$ of its output.

\begin{proposition}\label{prop:coherent}
For $j\leq n\leq d$, we have $\ip{v_\lambda}{M_{aa}v_\lambda}=0$ for $a<j$, and
\begin{equation}\label{eq:cohvalue}
\sum_{a=j}^n\ip{v_\lambda}{M_{aa}v_\lambda}=\kappa_n .
\end{equation}
\end{proposition}

\begin{proof}
By \eqref{eq:gram}, $\ip{v_\lambda}{M_{aa}v_\lambda}=\|P(v_\lambda\otimes e_a)\|^2$. For $a<j$, the weight $\lambda+e_a$ exceeds $\lp$ in the dominance order, so it is not a weight of $V_{\lp}$, and $P(v_\lambda\otimes e_a)=0$.

For the sum, identify $V_{\lp}$ with the range of $P$, and let $L:=L_n(\lambda)$ and $L^+:=L_n(\lp)$, the latter generated by $Tv_{\lp}$. The weights of $L\otimes\C^n$ agree with $\lambda$, and hence with $\lp$, beyond the $n$-th coordinate. Since $P$ preserves weights, Lemma~\ref{lem:index} shows that $P$ maps $L\otimes\C^n$ into $L^+$. Conversely, $L^+\subseteq L\otimes\C^n$, because $Tv_{\lp}=\sum_{a\leq j}x_a\otimes e_a$ with $x_a\in L$ by Lemma~\ref{lem:hw}. So $P$ restricts to the orthogonal projection of $L\otimes\C^n$ onto $L^+$, and $\sum_{a\leq n}M_{aa}$ restricts to its partial trace over $\C^n$, an operator on $L$. This operator commutes with $U(n)$ and has trace $D_n(\lp)$, so it equals $\kappa_n\one_L$ by Schur's lemma. Its expectation in $v_\lambda$ gives \eqref{eq:cohvalue}.
\end{proof}

Proposition~\ref{prop:coherent} and Lemma~\ref{lem:weyl} give \eqref{eq:target2}: $\kappa_dg_j=\kappa_j>0$, and $\kappa_dg_a=\kappa_a-\kappa_{a-1}>0$ for $a>j$, because $\kappa_a/\kappa_{a-1}=(\omega_j-\omega_a+1)/(\omega_j-\omega_a)>1$ by \eqref{eq:weyl}.

\subsection{Proof of Theorem~\ref{cor:adding}}\label{sec:assembly}

The two reductions bring the diagram into the shape \eqref{eq:shape}, where Proposition~\ref{prop:casimir} applies and Lemma~\ref{lem:weyl} identifies the bound. The dimension ratios then carry the bound back to the original diagram.

\begin{proof}[Proof of \eqref{eq:target}]
Let $n:=j+k-1\leq d$, let $\tilde\lambda$ be the diagram of Corollary~\ref{cor:cap}, and write $\tilde\lambda=\alpha+\beta$, where $\alpha$ consists of the columns of height $\geq j+k$. Then Proposition~\ref{prop:tall} applies, and $\beta+e_j$ satisfies both hypotheses in \eqref{eq:shape}; let $m_\beta$ be the column of its added box. Hence
\begin{align*}
\frac{\|M^\lambda_I\|}{\kappa_d(\lambda)}
&\ \overset{(1)}{\leq}\ \frac{\|M^{\tilde\lambda}_I\|}{\kappa_d(\tilde\lambda)}
\ \overset{(2)}{\leq}\ \frac{\|M^{\beta}_I\|}{\kappa_d(\tilde\lambda)}
\ \overset{(3)}{\leq}\ \frac{\kappa_d(\beta)}{\kappa_d(\tilde\lambda)}\cdot\frac{m_\beta+k-1}{m_\beta+d-j}\\
&\ \overset{(4)}{=}\ \frac{\kappa_n(\beta)}{\kappa_d(\tilde\lambda)}
\ \overset{(5)}{=}\ \frac{\kappa_n(\tilde\lambda)}{\kappa_d(\tilde\lambda)}
\ \overset{(6)}{=}\ \frac{\kappa_n(\lambda)}{\kappa_d(\lambda)},
\end{align*}
by (1) Corollary~\ref{cor:cap}, (2) Proposition~\ref{prop:tall}, (3) Proposition~\ref{prop:casimir}, and (4)--(6) Lemma~\ref{lem:weyl}(c), (b) and (a): (b) applies because $\beta$ arises from $\tilde\lambda$ by deleting columns of height $\geq j+k>n$, and (a) because $\tilde\lambda$ and $\lambda$ agree in the rows $\geq j$. That is, $\|M^\lambda_I\|\leq\kappa_n(\lambda)$, which is \eqref{eq:target}.
\end{proof}

\begin{proof}[Proof of Theorem~\ref{cor:adding}]
By Lemma~\ref{lem:reduction}, $\Psi_+(\sigma_0)$ is diagonal with entries $g_a$, so by \eqref{eq:target2} it has exactly $d-j+1$ nonzero eigenvalues. By Lemma~\ref{lem:reduction} and \eqref{eq:target}, $s_k(\Psi_+(\sigma))\leq\kappa_{j+k-1}/\kappa_d$ for every density matrix $\sigma$ on $V_\lambda$ and every $k\leq d-j+1$. For $\sigma=\sigma_0$ equality holds, since the entries $g_j,\dots,g_{j+k-1}$ already add up to $\kappa_{j+k-1}/\kappa_d$ by \eqref{eq:target2}. This is \eqref{eq:kappasums}. For $k=d-j+1$ it gives $1$, so $s_k(\Psi_+(\sigma_0))=1$ for all larger $k$ as well, and $\Psi_+(\sigma)\mj\Psi_+(\sigma_0)$. For $\Phi_+$, write $\sigma=\sum_ip_i|\phi_i\rangle\langle\phi_i|$ and use that $\Phi_+$ and $\Psi_+$ have the same nonzero output spectrum on pure states:
\[
s_k(\Phi_+(\sigma))\leq\sum_ip_i\,s_k\big(\Phi_+(|\phi_i\rangle\langle\phi_i|)\big)=\sum_ip_i\,s_k\big(\Psi_+(|\phi_i\rangle\langle\phi_i|)\big)\leq s_k(\Psi_+(\sigma_0))=s_k(\Phi_+(\sigma_0)).\qedhere
\]
\end{proof}

\subsection{Duality, and proof of Theorem~\ref{thm:main}}\label{sec:complement}

Conjugating every space in sight turns the pair $(\lambda,\lp)$ into the dual pair and exchanges removing with adding. This is a direct consequence of $\overline{V_\mu}\cong V_{\mu^c}\otimes{\det}^{-R}$; at the level of Clebsch--Gordan isometries, the corresponding symmetry under passing to dual representations is~\cite[Prop.~18]{MT25}. We include the short proof for completeness. Theorem~\ref{thm:main} is then Theorem~\ref{cor:adding} for the dual pair, once the dimension ratios are rewritten as hook lengths.

\begin{lemma}[Duality]\label{lem:complement}
Let $R\geq\lp_1$, and let $\Phi_+',\Psi_+'$ be the channels \eqref{eq:adding} of the pair $\big((\lp)^c,\lambda^c\big)$, whose diagrams differ by a box in row $j':=d+1-j$. There is an antiunitary $K:V_{\lp}\to V_{(\lp)^c}$ that maps coherent states to coherent states, such that $\Phi_-(\rho)$ and $\Phi_+'(K\rho K^{-1})$ have the same spectrum for every density matrix $\rho$ on $V_{\lp}$, and so do $\Psi_-(\rho)$ and $\Psi_+'(K\rho K^{-1})$.
\end{lemma}

\begin{proof}
For a unitary representation $V$, let $\overline V$ be the same space with the conjugate scalar multiplication, so that the identity map $V\to\overline V$ is antiunitary and equivariant, and the weights of $\overline V$ are those of $V$ negated. Then $\overline{V_\mu}\cong V_{\mu^c}\otimes{\det}^{-R}$, since both have the highest weight $(-\mu_d,\dots,-\mu_1)$~\cite[Lecture~15]{FultonHarris}; for $V=\Cd$, the conjugate $\overline V$ is the representation $\Cdbar$ of Section~\ref{sec:channels}, on which $U$ acts by $\overline U$. Viewed as a map $\overline{V_{\lp}}\to\overline{V_\lambda}\otimes\Cdbar$, the isometry $T$ becomes an isometric intertwiner $V_{(\lp)^c}\to V_{\lambda^c}\otimes\Cdbar$, the two determinant twists cancelling. Since $\lambda^c=(\lp)^c+e_{j'}$, this is the isometry \eqref{eq:S} of the pair $\big((\lp)^c,\lambda^c\big)$, which we call $S'$; it is unique up to a phase by Pieri's rule, so tracing out one factor or the other gives $\Phi_+'$ and $\Psi_+'$. Write $K$, $K'$, $K_0$ for the antiunitary identifications of $V_{\lp}$, $V_\lambda$, $\Cd$ with $V_{(\lp)^c}$, $V_{\lambda^c}$, $\Cdbar$. What we have shown is $S'K=(K'\otimes K_0)T$, and hence $\Phi_+'(K\rho K^{-1})=K'\Phi_-(\rho)K'^{-1}$ and $\Psi_+'(K\rho K^{-1})=K_0\Psi_-(\rho)K_0^{-1}$; conjugation by an antiunitary preserves spectra. Finally, $Kv_{\lp}$ has weight $(R,\dots,R)-\lp$ in $V_{(\lp)^c}$ (its weight in $\overline{V_{\lp}}$ is $-\lp$), which is the lowest weight of $V_{(\lp)^c}$. A permutation matrix maps it to the highest weight vector, and since $K$ is equivariant up to antilinearity, $K$ maps the coherent states of $V_{\lp}$ onto those of $V_{(\lp)^c}$.
\end{proof}

\begin{lemma}\label{lem:translate}
Let $\kappa'$ be \eqref{eq:kappa} for the pair $\big((\lp)^c,\lambda^c\big)$, and let $j'=d+1-j$. Then, for $1\leq k\leq j$,
\[
\frac{\kappa'_{j'+k-1}}{\kappa'_d}=\prod_{i=1}^{j-k}\Big(1-\frac{1}{h_{(i,m)}}\Big),
\]
with hook lengths taken in $\lp$.
\end{lemma}

\begin{proof}
Put $n:=j'+k-1=d-j+k$, and let $\lambda_\downarrow$ and $\lp_\downarrow$ be the diagrams formed by the last $n$ rows of $\lambda$ and of $\lp$. The first $n$ rows of $\lambda^c$ are $(R-\lambda_d,\dots,R-\lambda_{d+1-n})$, so the $U(n)$-representation with that highest weight is, up to a determinant twist, dual to the one with highest weight $\lambda_\downarrow$. Hence $D_n(\lambda^c)=D_n(\lambda_\downarrow)$, and likewise for $\lp$; for $n=d$ these are $\dim V_\lambda$ and $\dim V_{\lp}$. Since $\kappa'_n=D_n(\lambda^c)/D_n((\lp)^c)$, this gives
\[
\frac{\kappa'_n}{\kappa'_d}=\frac{\dim V_{\lp}}{\dim V_\lambda}\Big/\frac{D_n(\lp_\downarrow)}{D_n(\lambda_\downarrow)},
\]
and it remains to compute the two ratios. Both follow from the hook-content formula $D_n(\mu)=\prod_{(a,b)\in\mu}(n+b-a)/h_{(a,b)}$, in which $b-a$ is the content of the box $(a,b)$~\cite[Cor.~7.21.4]{Stanley}. Adding the box $(j,m)$ to $\lambda$ contributes its content $m-j$ and hook length $1$, raises by one the hook lengths in its row and in its column, and changes nothing else. In $\lp_\downarrow$ the box sits in row $k$; its row survives entire, and of its column only the rows $j-k+1,\dots,j-1$ survive. Therefore
\[
\frac{D_n(\lp_\downarrow)}{D_n(\lambda_\downarrow)}=(d+m-j)\prod_{c<m}\Big(1-\frac{1}{h_{(j,c)}}\Big)\prod_{i=j-k+1}^{j-1}\Big(1-\frac{1}{h_{(i,m)}}\Big),
\]
where we used $n+(m-k)=d+m-j$. Taking $n=d$ and $k=j$ gives the same expression for $\dim V_{\lp}/\dim V_\lambda$, with the last product running over all $i<j$. In the quotient, the prefactor and the row product cancel, and $\prod_{i=1}^{j-k}(1-1/h_{(i,m)})$ remains.
\end{proof}

\begin{proof}[Proof of Theorem~\ref{thm:main}]
By Lemma~\ref{lem:complement}, the channels $\Phi_-,\Psi_-$ have the same output spectra as the channels $\Phi_+',\Psi_+'$ of the dual pair $\big((\lp)^c,\lambda^c\big)$, with coherent states corresponding to coherent states, and the box is added there in row $j'=d+1-j$. Theorem~\ref{cor:adding} gives the majorization. The coherent-state output of $\Psi_+'$ is diagonal with $d-j'+1=j$ nonzero eigenvalues and $s_k=\kappa'_{j'+k-1}/\kappa'_d$, which is \eqref{eq:hooks} by Lemma~\ref{lem:translate}. Finally, $\Psi_-(\rho_0)$ is diagonal in the basis $(e_a)$ because the vectors $T_av_{\lp}$ have distinct weights, and $\Phi_-(\rho_0)$ has the same nonzero eigenvalues because $\Phi_-$ and $\Psi_-$ are complementary. It remains to identify the individual entries \eqref{eq:entries}. By the proof of Lemma~\ref{lem:complement}, $\Psi_+'(K\rho_0K^{-1})=K_0\Psi_-(\rho_0)K_0^{-1}$, and $K_0$ preserves diagonal entries. Moreover, $Kv_{\lp}$ is a lowest weight vector of $V_{(\lp)^c}$, and the permutation matrix $w$ with $we_a=e_{d+1-a}$ maps it to a highest weight vector, up to a phase. By covariance, the entry of $\Psi_-(\rho_0)$ at $e_a$ is therefore the entry $g'_{d+1-a}$ of the coherent-state output of $\Psi_+'$, given by \eqref{eq:target2} for the dual pair. This vanishes for $a>j$, since then $d+1-a<j'$. For $a\leq j$, \eqref{eq:target2} gives $\kappa'_dg'_{j'}=\kappa'_{j'}$ and $\kappa'_dg'_{b}=\kappa'_{b}-\kappa'_{b-1}$ for $b>j'$, and Lemma~\ref{lem:translate}, applied with $k=j-a+1$ and, if $a<j$, with $k=j-a$, turns this into
\[
g'_{d+1-a}=\prod_{i<a}\Big(1-\frac{1}{h_{(i,m)}}\Big)-\prod_{i\leq a}\Big(1-\frac{1}{h_{(i,m)}}\Big)=\frac{1}{h_{(a,m)}}\prod_{i<a}\Big(1-\frac{1}{h_{(i,m)}}\Big),
\]
where for $a=j$ the second product is $0$, because $h_{(j,m)}=1$. This is \eqref{eq:entries}.
\end{proof}

\subsection{Entropy and reduced density matrices}\label{sec:corproofs}

\begin{proof}[Proof of Corollary~\ref{cor:entropy}]
$A\mj B$ implies $\Tr f(A)\geq\Tr f(B)$ for all concave $f$ \cite[Ch.~3]{MOA}.
\end{proof}

For the $N$-particle results, fix a removable box $(j,m)$ of $\nu$ and let $R_{(j,m)}$ be the projection of $\otimes^{N-1}\Cd$ onto its isotypic component $V_{\nu-e_j}\otimes S_{\nu-e_j}$. Then $R_{(j,m)}\otimes\one$ commutes with $U(d)$ and with $S_{N-1}$. It therefore preserves $V_\nu\otimes S_\nu$, and there it is the projection $P_{(j,m)}$ onto the $S_{N-1}$-isotypic component $V_\nu\otimes S_\nu^{(j,m)}$.

\begin{lemma}\label{lem:nestedspace}
Inside $\big(V_{\nu-e_j}\otimes S_{\nu-e_j}\big)\otimes\Cd=\big(V_{\nu-e_j}\otimes\Cd\big)\otimes S_{\nu-e_j}$, we have
\[
V_\nu\otimes S_\nu^{(j,m)}=T(V_\nu)\otimes S_{\nu-e_j},
\]
where $T$ is the isometry \eqref{eq:T} for $\lp=\nu$ and the box $(j,m)$. Hence $\gamma=\Psi_-(\rho)$ for every density matrix $\varrho$ on $V_\nu\otimes S_\nu^{(j,m)}$, where $\rho$ is its $V_\nu$-marginal.
\end{lemma}

\begin{proof}
Both sides are the intersection of the range of $R_{(j,m)}\otimes\one$ with the $\nu$-isotypic component $V_\nu\otimes S_\nu$ of $\otimes^N\Cd$. For the left-hand side this is the definition of $P_{(j,m)}$. For the right-hand side, $U(d)$ acts on $(V_{\nu-e_j}\otimes\Cd)\otimes S_{\nu-e_j}$ through the first factor only, whose $\nu$-isotypic component is $T(V_\nu)$ by Pieri's rule. Now write $\varrho=(T\otimes\one)\varrho'(T\otimes\one)^\dagger$ with $\varrho'$ on $V_\nu\otimes S_{\nu-e_j}$, and trace out $V_{\nu-e_j}\otimes S_{\nu-e_j}$: this gives $\gamma=\Tr_{V_{\nu-e_j}}(T\rho T^\dagger)=\Psi_-(\rho)$.
\end{proof}

\begin{proof}[Proof of Theorem~\ref{thm:nested}]
By Lemma~\ref{lem:nestedspace}, the $\gamma$ that occur are exactly the outputs of $\Psi_-$, since every density matrix on $V_\nu$ is the marginal of some $\varrho$. Theorem~\ref{thm:main} gives $\spec\gamma\mj p^{(j,m)}$, with equality for $\rho=\rho_0$, which is the marginal of $|v_\nu\rangle\langle v_\nu|\otimes\sigma$. Conversely, the set of outputs is convex, and it is stable under conjugation by $U(d)$, since $\Psi_-(\pi_{\lp}(U)\rho\,\pi_{\lp}(U)^\dagger)=U\Psi_-(\rho)U^\dagger$. If $\spec A\mj p^{(j,m)}$, then $\spec A$ is a convex combination of permutations of $p^{(j,m)}$, by Rado's theorem~\cite[Ch.~4]{MOA} (the Hardy--Littlewood--P\'olya theorem combined with Birkhoff's theorem). So $A$ is a convex combination of conjugates of $\Psi_-(\rho_0)$, and it is an output.
\end{proof}

\begin{proof}[Proof of Corollary~\ref{thm:rdm}]
Since the $P_{(j,m)}$ add up to the identity on $V_\nu\otimes S_\nu$, we have $\varrho=\sum_{(j,m),(j',m')}P_{(j,m)}\varrho P_{(j',m')}$. The cross terms do not contribute to $\gamma$: the partial trace over the first $N-1$ factors is cyclic in these factors, so for two different boxes
\[
\Tr_{1,\dots,N-1}\big((R_{(j,m)}\otimes\one)\,\varrho\,(R_{(j',m')}\otimes\one)\big)=\Tr_{1,\dots,N-1}\big((R_{(j',m')}R_{(j,m)}\otimes\one)\,\varrho\big)=0 .
\]
Hence
\begin{equation}\label{eq:gammadecomp}
\gamma=\sum_{(j,m)}c_{(j,m)}\gamma_{(j,m)},\qquad c_{(j,m)}:=\Tr P_{(j,m)}\varrho,
\end{equation}
where $\gamma_{(j,m)}$ is the reduced density matrix of the state $P_{(j,m)}\varrho P_{(j,m)}/c_{(j,m)}$ on $V_\nu\otimes S_\nu^{(j,m)}$ (terms with $c_{(j,m)}=0$ are omitted). The bound follows from Theorem~\ref{thm:nested} and the convexity of $s_k$, and the states of Theorem~\ref{thm:nested}, which lie in $V_\nu\otimes S_\nu$, attain it.
\end{proof}

\begin{proof}[Proof of Corollary~\ref{thm:polytope}]
By \eqref{eq:gammadecomp} and Theorem~\ref{thm:nested}, the $\gamma$ that occur are the convex combinations $\sum_ic_iA_i$ with $\spec A_i\mj p^{(i)}$, where $i$ runs over the removable boxes; every such combination occurs, for a mixture of states on the individual summands. For these, $\spec\gamma\mj\sum_ic_i\spec A_i\mj\sum_ic_ip^{(i)}$, with spectra in decreasing order. For the converse, the permutohedron $\mathcal P(x)=\{y:y\mj x\}$ satisfies $\mathcal P(\sum_ic_ip^{(i)})=\sum_ic_i\mathcal P(p^{(i)})$, because both sides are convex with the same support function: for $x$ in decreasing order, $\mathcal P(x)$ is the convex hull of the permutations of $x$ by Rado's theorem~\cite[Ch.~4]{MOA}, so its support function is $y\mapsto\sum_sx_sy^\downarrow_s$ by the rearrangement inequality, and this is additive in $x$. So if \eqref{eq:polytope} holds, then $\spec\gamma=\sum_ic_iq^{(i)}$ with $q^{(i)}\mj p^{(i)}$, and $A_i:=W\operatorname{diag}(q^{(i)})W^\dagger$, with $W$ a unitary diagonalizing $\gamma$, does the job. The same identity shows that $\bigcup_c\mathcal P\big(\sum_ic_ip^{(i)}\big)=\operatorname{conv}\bigcup_i\mathcal P(p^{(i)})$, the convex hull of all permutations of the $p^{(i)}$; the spectra are its intersection with the cone $x_1\geq\dots\geq x_d$.
\end{proof}

\subsection{The minimizers}\label{sec:minimizers}

States supported on $\mathcal S$ are optimal by a direct computation (Lemma~\ref{lem:Sout}). Conversely, if a unit vector $\phi$ has $\Psi_-(|\phi\rangle\langle\phi|)=\Psi_-(\rho_0)$, then $T_a^{\dagger}T_a\phi=\xi_a\phi$ for every $a\leq j$, and these conditions confine $\phi$ to $\mathcal S$ one row at a time (Lemma~\ref{lem:levels}). In the tableau picture of Section~\ref{sec:consequences}, $\phi$ thus ends up in the span of the tableaux whose row $a$ is filled with the entry $a$ for every $a\leq j$. Here $\xi_a=\|T_av_{\lp}\|^2$ are the entries \eqref{eq:entries}, so that $\Psi_-(\rho_0)=\operatorname{diag}(\xi_1,\dots,\xi_d)$.

\begin{lemma}\label{lem:Sout}
Every density matrix $\sigma$ supported on $\mathcal S$ satisfies $\Psi_-(\sigma)=\Psi_-(\rho_0)$.
\end{lemma}

\begin{proof}
Covariance says that $T_a\pi_{\lp}(U)=\pi_\lambda(U)\sum_bU_{ab}T_b$ for all $U\in U(d)$. For $U$ fixing $e_1,\dots,e_j$ this gives $T_a\pi_{\lp}(U)=\pi_\lambda(U)T_a$ for $a\leq j$, so $T_b^{\dagger}T_a$ commutes with $\pi_{\lp}(G_j)$ whenever $a,b\leq j$. As $\mathcal S$ is irreducible, $T_b^{\dagger}T_a$ compressed to $\mathcal S$ is a scalar, and evaluating at $v_{\lp}$, for which $T_b^{\dagger}T_av_{\lp}$ has weight $\lp+e_b-e_a$, shows that it is $\delta_{ab}\xi_a$. Since $\sum_{a\leq j}\xi_a=1=\sum_a\|T_a\phi\|^2$, every unit vector $\phi\in\mathcal S$ has $T_a\phi=0$ for $a>j$, and hence $\ip{T_b\phi}{T_a\phi}=\delta_{ab}\xi_a$ for all $a,b$. That is, $\Psi_-(|\phi\rangle\langle\phi|)=\Psi_-(\rho_0)$ for every unit vector $\phi\in\mathcal S$, and the lemma follows by linearity.
\end{proof}

For $0\leq r\leq j$, let $G_r\subseteq U(d)$ be the subgroup fixing $e_{j-r+1},\dots,e_j$, and let $\mathcal S_r\subseteq V_{\lp}$ be the $G_r$-submodule generated by $v_{\lp}$, with projection $P_{\mathcal S_r}$; for $r=j$ this is the group $G_j$ of Section~\ref{sec:consequences}. It is irreducible, with highest weight $\lp$ with the rows $j-r+1,\dots,j$ deleted; $\mathcal S_0=V_{\lp}$ and $\mathcal S_j=\mathcal S$.

\begin{lemma}\label{lem:levels}
For $1\leq r\leq j$ and $a=j-r+1$, the largest eigenvalue of $P_{\mathcal S_{r-1}}T_a^{\dagger}T_a$ on $\mathcal S_{r-1}$ is $\xi_a$, and the corresponding eigenspace is $\mathcal S_r$.
\end{lemma}

\begin{proof}
The operator commutes with $G_r$, so it acts by a scalar on each $G_r$-component of $\mathcal S_{r-1}$. We compute these scalars and show that the largest is $\xi_a$, attained only on $\mathcal S_r$. Each scalar is a combination of known ones for single pairs of diagrams, with nonnegative weights that do not depend on the component (Step~1). The weights are fixed by the scalars on a few special components (Step~2). On a face of the set of components, which contains every maximizer, the combination is a single product (Step~3), and this product is maximal exactly at $\mathcal S_r$ (Step~4).

Let $\mu$ be the highest weight of $\mathcal S_{r-1}$, with its rows numbered consecutively, so that its rows $1,\dots,a$ are those of $\lp$. Put $\ell_c=\mu_c-c$, and let $\ell_\ast=m-j$ be the content of the removed box, so that $h_{(c,m)}=\ell_c-\ell_\ast+1$ for $c\leq a$. The indices $s,t$ run over the removable rows of $\mu$; row $a$ is one of them, since $\mu_a\geq m>\lp_{j+1}$.

For $r=1$ we have $P_{\mathcal S_0}=\one$, only $t=j$ occurs in Step~1, with $w_j=1$, and $\ell_j=\ell_\ast$; so Step~2 is empty, Step~3 reduces to the cancellation on the face, and the lemma reduces to maximizing the formula of~\cite{RGKZ26} (Step~4); the maximum $\xi_j$ is the product in \eqref{eq:hook1} for the box $(j,m)$, found in~\cite{Reuvers19}. The work is in $r\geq2$, where $T_a^{\dagger}T_a$ does not preserve $\mathcal S_{r-1}$.

\emph{Step 1: the scalars as a fixed combination.} Under $G_r$, the module $\mathcal S_{r-1}$ branches without multiplicity into components whose highest weights $\eta$ interlace $\mu$~\cite[\S8.1]{GoodmanWallach}. Let $\chi(\eta)$ be the scalar by which our operator acts on the component $\eta$, and put $\zeta_b=\eta_b-b$. For a single pair $(\mu-e_t,\mu)$, with isometric intertwiner $T^{(t)}$, the corresponding scalar of $T^{(t)\dagger}_aT^{(t)}_a$ is known:
\[
\chi_t(\eta)=\prod_b(\zeta_b-\ell_t)\Big/\prod_{c\neq t}(\ell_c-\ell_t+1)\ \geq0,
\]
by the classical formula for squares of reduced Wigner coefficients~\cite{VilenkinKlimyk}, in the form of~\cite[Thm.~1, Eq.~(14)]{RGKZ26}, with $\mu-e_t$, $\eta$, $t$ in place of their $\mu$, $\nu$, $k$. The formula is stated there for $T^{(t)}_aT^{(t)\dagger}_a$, which has the same nonzero scalars because $T^{(t)}_a$ is $G_r$-equivariant; its numerator vanishes exactly when $\eta$ does not interlace $\mu-e_t$. To relate our operator to these, let $I$ be the set of coordinates moved by $G_{r-1}$. The map $\phi\mapsto\sum_{b\in I}T_b\phi\otimes e_b$ is a $G_{r-1}$-intertwiner from $\mathcal S_{r-1}$ into $V_\lambda\otimes\operatorname{span}\{e_b:b\in I\}$. Decompose $V_\lambda$ under $G_{r-1}$: by Pieri's rule, the map lands in the isotypic components of type $\mu-e_t$, and there it is $T^{(t)}$ tensored with a vector $u_t$ of the multiplicity space. These components are orthogonal, so
\[
\chi(\eta)=\sum_tw_t\chi_t(\eta),\qquad w_t=\|u_t\|^2\geq0,
\]
with weights that do not depend on $\eta$.

\emph{Step 2: the weights are fixed by the deletions.} Let $\eta^{(s)}$ be $\mu$ with row $s$ deleted, let $s'$ be the index in $\{1,\dots,d\}$ of that row, and let $\tau\in G_{r-1}$ exchange $e_a$ and $e_{s'}$. Conjugation by $\tau$ turns $G_r$ into the subgroup fixing $e_{s'}$ instead of $e_a$, under which $v_{\lp}$ generates a module of highest weight $\eta^{(s)}$. So $\pi_{\lp}(\tau)v_{\lp}$ is a unit vector in the component $\eta^{(s)}$, and by covariance and Lemma~\ref{lem:hw}
\begin{equation}\label{eq:deletions}
\chi(\eta^{(s)})=\|T_{s'}v_{\lp}\|^2=\begin{cases}\xi_s,& s\leq a,\\ 0,& s>a.\end{cases}
\end{equation}
Since $\chi_t(\eta^{(s)})$ vanishes for $s>t$ but not for $s=t$, the system \eqref{eq:deletions} is triangular in the weights: any combination $\sum_t\alpha_t\chi_t$ with the values \eqref{eq:deletions} has $\alpha_t=w_t$.

\emph{Step 3: a single product on a face.} Call the set of $\eta$ with $\eta_b=\mu_{b+1}$ for all $b\geq a$ the face. We show that $w_t=0$ for $t>a$, and that on the face
\[
\chi(\eta)=F(\eta):=\Delta(\ell_\ast)\Big/\prod_{c\leq a}h_{(c,m)},\qquad \Delta(y):=\prod_{b<a}(\zeta_b-y).
\]
On the face, $\zeta_b-\ell_t=\ell_{b+1}-\ell_t+1$ for $b\geq a$, so these factors of $\chi_t$ cancel the factors of its denominator with $c>a$, and for $t\leq a$
\[
\chi_t(\eta)=\frac{\Delta(\ell_t)}{\prod_{c\leq a,\,c\neq t}(\ell_c-\ell_t+1)}.
\]
As $\Delta$ has degree $a-1$, Lagrange interpolation at the nodes $\ell_1,\dots,\ell_a$ gives $F=\sum_{t\leq a}\alpha_t\chi_t$ on the face, with coefficients $\alpha_t$ that depend only on the $\ell_c$ and $\ell_\ast$, not on $\eta$. Only removable rows $t$ occur: if $\mu_t=\mu_{t+1}$, interlacing forces $\eta_t=\mu_t$, so $\Delta(\ell_t)=0$. The combination $\sum_{t\leq a}\alpha_t\chi_t$ has the values \eqref{eq:deletions}. Indeed, the deletions $\eta^{(s)}$ with $s\leq a$ lie on the face, where it equals $F(\eta^{(s)})=\xi_s$ by \eqref{eq:entries}; and for $s>a$ it vanishes, since $\chi_t(\eta^{(s)})=0$ for $t\leq a<s$. By Step~2, $\alpha_t=w_t$. So $w_t=0$ for $t>a$, and $\chi=F$ on the face.

\emph{Step 4: the maximum.} By Step~3, $\chi=\sum_{t\leq a}w_t\chi_t$. For $b\geq a\geq t$ the factor $\zeta_b-\ell_t$ of $\chi_t$ is nonpositive and largest in modulus at $\eta_b=\mu_{b+1}$. So moving $\eta_b$ to $\mu_{b+1}$ for all $b\geq a$, which keeps $\eta$ interlacing $\mu$ and leaves the other factors unchanged, does not decrease any $\chi_t(\eta)$, $t\leq a$, and increases $\chi(\eta)$ strictly if $\chi(\eta)>0$ and $\eta$ is not already on the face. Every maximizer of $\chi$ therefore lies on the face. There $\chi=F$ is a product of positive factors, increasing in each $\zeta_b$ with $b<a$, and so largest only at $\eta_b=\mu_b$ $(b<a)$. That is $\eta^{(a)}$, the highest weight of $\mathcal S_r$, and $F(\eta^{(a)})=\xi_a$.
\end{proof}

\begin{proof}[Proof of Theorem~\ref{thm:minimizers}]
By Lemma~\ref{lem:Sout} and covariance, the states in the theorem have output spectrum $p^{(j,m)}$. Conversely, let $\rho=\sum_iq_i|\phi_i\rangle\langle\phi_i|$ with $q_i>0$ have this output spectrum, and take $f$ strictly concave. As $A\mapsto\Tr f(A)$ is strictly concave~\cite{Carlen10}, equality in
\[
\Tr f\big(\Psi_-(\rho)\big)\ \geq\ \sum_iq_i\,\Tr f\big(\Psi_-(|\phi_i\rangle\langle\phi_i|)\big)\ \geq\ \Tr f\big(\Psi_-(\rho_0)\big),
\]
where the second inequality is Corollary~\ref{cor:entropy}, forces all $\Psi_-(|\phi_i\rangle\langle\phi_i|)$ to equal $\Psi_-(\rho)$. Their common value has spectrum $p^{(j,m)}$, so after replacing $\rho$ by $\pi_{\lp}(U)^\dagger\rho\,\pi_{\lp}(U)$ for a suitable $U$ we may assume, by covariance, that it is $\Psi_-(\rho_0)$. It remains to show that each $\phi=\phi_i$ lies in $\mathcal S$. For $N_k=\sum_{a=j-k+1}^{j}T_a^{\dagger}T_a$ and a unit vector $\psi$, the number $\ip{\psi}{N_k\psi}$ is a sum of $k$ diagonal entries of $\Psi_-(|\psi\rangle\langle\psi|)$, hence at most the $k$-th partial sum of $p^{(j,m)}$ by Theorem~\ref{thm:main}, and $\phi$ attains this by \eqref{eq:entries}. So $\phi$ is an eigenvector of each $N_k$ for its largest eigenvalue, and subtracting the eigenvalue equations for $N_k$ and $N_{k-1}$ gives $T_a^{\dagger}T_a\phi=\xi_a\phi$ for $a\leq j$. If $\phi\in\mathcal S_{r-1}$, then $P_{\mathcal S_{r-1}}T_a^{\dagger}T_a\phi=\xi_a\phi$ with $a=j-r+1$, so $\phi\in\mathcal S_r$ by Lemma~\ref{lem:levels}. Starting from $\mathcal S_0=V_{\lp}$, this gives $\phi\in\mathcal S$.
\end{proof}

\begin{proof}[Proof of Corollary~\ref{cor:coherent}]
We first show that a unit vector $\phi\in\mathcal S$ is coherent in $V_{\lp}$ if and only if it is coherent in the $G_j$-module $\mathcal S$, that is, lies in the $G_j$-orbit of $v_{\lp}$ up to a phase. The weights of $\mathcal S$ agree with $\lp$ in their first $j$ entries, so for $\phi\in\mathcal S$ the matrix $\big(\ip{\phi}{E_{ab}\phi}\big)_{a,b}$ is $\operatorname{diag}(\lp_1,\dots,\lp_j)\oplus M'$, where $M'$ is the same matrix for $\phi$ regarded as a vector in $\mathcal S$. A unit vector is coherent exactly when this matrix has maximal Hilbert--Schmidt norm, namely the norm of the highest weight~\cite{DF77,Perelomov}, and the block structure shows that this happens in $V_{\lp}$ if and only if it happens in $\mathcal S$.

Now (a) follows from Theorem~\ref{thm:minimizers}: if $\dim\mathcal S=1$, every state supported on $\mathcal S$ is $\rho_0$, and otherwise $\mathcal S$ carries non-coherent states, for instance the normalized identity. For (b), every pure minimizer is coherent exactly when $G_j$ acts transitively on the unit sphere of $\mathcal S$ up to phase. By the classification of compact groups transitive on spheres~\cite{MS43,Borel49}, this happens exactly when $\mathcal S$ is one-dimensional or, up to a determinant twist, $\C^{d-j}$ or its dual: the image of the subgroup $SU(d-j)\subseteq G_j$ is a quotient of $SU(d-j)$, and the only such group in the classification is $SU(D)$ in its defining representation, $D=\dim\mathcal S$. In terms of the diagram, this means that $(\lp_{j+1},\dots,\lp_d)$ is constant, $(c+1,c,\dots,c)$ or $(c,\dots,c,c-1)$.
\end{proof}

\section*{Acknowledgements}

I thank Tommaso Aschieri, Dmitry Grinko, B\l{}a\.zej Ruba, Jan Philip Solovej and Elias Theil for interesting discussions. This paper was written while I was visiting the Isaac Newton Institute for Mathematical Sciences, Cambridge, during the programme \emph{Mathematics of Many-Body Entanglement}, and I thank the Institute for its support and hospitality. This work was partially supported by a grant from the Simons Foundation, through a Simons Foundation Fellowship of the Institute, and by EPSRC grant EP/Z000580/1. I thank Nilanjana Datta for hosting me during a previous research visit to Cambridge, and the Royal Society for supporting that visit (AL\textbackslash 24100022). I also thank QMATH at the University of Copenhagen for hospitality during a visit. I am a member of the Gruppo Nazionale per la Fisica Matematica (GNFM) of the Istituto Nazionale di Alta Matematica Francesco Severi (INdAM).

\section*{Competing interests}

The author has no competing interests to declare.

\section*{Data availability}

No data were generated or analysed in this study. 

\section*{Declaration on the use of artificial intelligence}
The main result of this paper and the diagram reduction strategy are due to me; I first conjectured the result on the basis of numerical computations with the algorithm of~\cite{Reuvers17}. Claude provided the estimate in Section~\ref{sec:casimir}, obtained the results of Section~\ref{sec:minimizers} and drafted the paper, which I then edited and checked. The counterexamples in Remark~\ref{rem:cartan} were found with the help of AI. I take full responsibility for the content.

\end{document}